\documentclass[12pt, draftclsnofoot, onecolumn]{IEEEtran}
\ifCLASSINFOpdf
\else
\fi
\makeatletter
\def\ps@headings{
\def\@oddhead{\mbox{}\scriptsize\rightmark \hfil \thepage}%
\def\@evenhead{\scriptsize\thepage \hfil \leftmark\mbox{}}%
\def\@oddfoot{}%
\def\@evenfoot{}}
\makeatother 
\IEEEoverridecommandlockouts

\usepackage{cite}
\usepackage{amsmath,amssymb,amsfonts}
\usepackage{graphicx}
\usepackage{textcomp}
\usepackage{xcolor}
\usepackage{algorithm}
\usepackage{algorithmicx}
\usepackage{algpseudocode}
\usepackage{epstopdf}
\usepackage{subfigure}
\usepackage{caption}
\usepackage{stfloats}
\usepackage{hyperref}
\usepackage{amsmath}
\usepackage{float}
\begin{document}
\title{ Effective Capacity of  URLLC over Parallel Fading Channels with Imperfect Channel State Information}
\author{\IEEEauthorblockN{Hongsen  Peng, \IEEEmembership{Graduate Student Member,~ IEEE} \\
Meixia Tao,\IEEEmembership{ Fellow,~IEEE}}

\thanks{\quad This work is supported by the National Natural Science Foundation of China under grant 61941106.}
\thanks{\quad The authors are with Department of Electronic Engineering, Shanghai Jiao Tong University, Shanghai, China. Emails: \{hs-peng, mxtao\}@sjtu.edu.cn. }

}

\maketitle

\begin{abstract}
This paper investigates the effective capacity of a point-to-point ultra-reliable low latency communication (URLLC) transmission over multiple parallel sub-channels at finite blocklength (FBL) with imperfect channel state information (CSI).
Based on reasonable assumptions and approximations, we derive the effective capacity as a function of  the pilot length, decoding error probability, transmit power and the sub-channel number. Then we reveal significant impact of the above parameters on the effective capacity.  A closed-form lower bound of the effective capacity is derived and an alternating optimization based algorithm is proposed to find the optimal pilot length and decoding error probability.
Simulation results validate our theoretical analysis and show that the closed-form lower bound is very tight. In addition, through the simulations of the optimized effective capacity, insights for pilot length and decoding error probability optimization  are provided  to evaluate the optimal parameters in realistic systems.
\end{abstract}
\begin{IEEEkeywords}
URLLC, parallel fading channels,  finite blocklength regime, imperfect CSI, effective capacity
\end{IEEEkeywords}
\newpage
\section{Introduction}
\label{Introduction}

Ultra-reliable and low-latency communication (URLLC) is one of the main service scenarios supported by 5G wireless networks and beyond. It can enable many mission-critical applications such as autonomous driving, industrial automation,  remote surgery and so on \cite{8472907}.
 URLLC entails much stricter quality-of-service (QoS) requirements including a small packet error probability between  $10^{-9}\sim10^{-5}$ and a very low latency around 1 ms \cite{8472907, 8705373, 8469808, 8452975}.
 For URLLC, the latency is defined as end-to-end delay \cite{8472907, 8705373}, which not only contains the transmission delay in the physical layer but also includes the  queueing delay in the buffer.
The reliability is defined as the probability that a finite-size data packet is successfully delivered by the transmitter to the receiver within a target  time duration \cite{8472907}.  Therefore, the reliability inherently includes the latency constraint and we can also call reliability as delay violation probability as it requires the latency constraint satisfied.
 Conventionally, due to the stochastic nature of wireless channels, it is quite challenging to support communications that have stringent requirement for delay and reliability.

To ensure low transmission latency, the packet size of URLLC is usually very small. Under the short-packet, or equivalently, finite blocklength (FBL), transmission constraint,  the conventional Shannon formula  derived from infinite blocklength  model with error free transmission   is not applicable \cite{7529226}. Specifically, transmission at FBL  will bring non-negligible loss in the achievable channel coding rate and non-zero probability of decoding error \cite{5452208}.
Another related and important problem is to study the impact of imperfect channel state information (CSI) under the FBL constraint \cite{8403963}. Training based CSI acquisition will bring the tradeoff  between the pilot length and the payload length.
On the one hand, longer pilot length will improve the accuracy of the channel estimation but reduce the remaining blocklength for payload transmission.  On the other hand, shorter pilot length will deteriorate the estimation of the channel while increase the blocklength for the payload transmission.
{The aim of this work is to investigate the throughput performance of URLLC from the link layer perspective in term of  effective capacity,  at FBL taking  both statistical delay requirements and imperfect CSI into consideration. }
\subsection{Related Works}

Information-theoretic study on the channel coding rate in the FBL regime was presented in \cite{5452208} recently, where the authors proposed an accurate approximation (i.e., normal approximation) of the achievable rate over additive white Gaussian noise (AWGN) channels as a function of the signal-to-noise  ratio (SNR), blocklength of the codeword, and decoding error probability.
The extension to multiple-antenna fading channels was considered in \cite{6802432,7362178}.
These new results have inspired many new research papers in different aspects.

Many of the existing works on URLLC have focused on how to achieve the stringent QoS requirement in the physical layer.
 Xu $et$~$ al.$  investigated the energy-efficient packet scheduling problem over quasi-static block fading channels \cite{7463506}. In \cite{8345745}, Sun $et$~$ al.$ analyzed a  nonorthogonal multiple access (NOMA) based downlink low latency transmission problem.  Hu $et$~$ al.$ \cite{8259329} investigated the relay based cooperative URLLC  transmission. Makki  $et$~$ al.$ \cite{6888474} studied the hybrid automatic repeat request  (ARQ) protocol in block fading channels. However, all of  these works assumed perfect CSI at both the transmitter and the receiver.
There are also many papers that consider both FBL and imperfect CSI. The works \cite{7996416,8885564,9013958}  investigated communications in block fading channels with pilot length optimizations under the  assumption of the imperfect CSI  estimation in the physical layer. Cheng $et$~$ al.$ \cite{9120794} investigated the  resource allocation problem for downlink  orthogonal frequency-division multiple-access (OFDMA) system with bounded CSI error  model. Ren $et$~$ al.$ \cite{9044874} studied the power allocation problem of the pilot and payload in massive MIMO system with imperfect CSI in the uplink.

According to the inherent relationship between delay and reliability, it is also essential to investigate the throughput performance of URLLC by considering delay violation probability  from the link layer perspective.
In general, there are two well known analysis tools  to study link layer performance, namely,  effective capacity and stochastic network calculus (SNC).
Effective capacity is a useful measure of system throughput with statistical QoS guarantee \cite{1210731} and has been widely studied, for example in {\cite{Amjad2019, 8913775}}.  It captures the maximum constant data arrival rate that a given service process can support under certain delay violation or buffer overflow probability, also known as,  statistical QoS constraints.
On the other hand, stochastic network calculus is used to characterize the non-asymptotic probabilistic performance bounds in terms of the distribution of fading channels and arrival processes \cite{6932489}, relaxing the intractable delay target violation probability to the tractable upper bound.

 Effective capacity in the FBL regime  was first analyzed by Gursoy in  \cite{GursoyDec.2013} with queueing constraints. It is proved that there exists a unique decoding error probability that maximizes the effective capacity.
 Hu $et$~$ al.$ \cite{8402240} investigated the optimal multiuser  power allocation problem to maximize the normalized sum effective capacity with fixed decoding error probability.
The authors in \cite{8645712} obtained a closed-form approximation of  the effective capacity under Rayleigh fading channel for machine type communications (MTC) through proper expansion and then investigated the  power-delay tradeoff for fixed effective capacity.
It is worthwhile to mention that, all the above works consider the channel coding performed only within one fading block. The following works consider the channel coding across multiple blocks.
 Choi \cite{8543235} studied the effective capacity of parallel multi-channel for low latency communication  for both infinite blocklength and FBL.  Therein, FBL is only considered for the statistical CSI case where the channel coding rate  remains unchanged for all fading blocks.
 Qiao $et$~$ al.$  \cite{QiaoAug.2019} investigated the effective capacity with FBL channel coding over multiple coherence blocks
 and  revealed the relationships between the decoding error probability, coherence block  number  and the effective capacity.
 Note that the tradeoffs among these system parameters are discussed via simulation results only.
Nevertheless,  none of these effective capacity related works considered the imperfect CSI scenario and finite blocklength at the same time.

The tool of SNC also has been applied for the link layer performance analysis for URLLC. Specifically, Xiao $et$~$ al.$\cite{8638930, 8649644} investigated the power allocation problem and  analyzed the delay performance  in the link layer in downlink NOMA systems utilizing SNC. But these two papers did not take FBL channel coding and imperfect CSI into consideration.
Schiessl $et$~$ al.$\cite{8421266, 8640115}  investigated the delay performance through a newly derived closed-form but approximate decoding error probability by taking both FBL and imperfect CSI into consideration with SNC. These two paper provided  rate adaption strategies which lead to a minimum delay violation probability.
{However, none of these above SNC related works derived explicit closed-form relation between the delay violation probability and the considered parameters.}

In this paper, we are primarily interested in the throughput performance with QoS guarantee instead of delay performance. Therefore, effective capacity is more suitable and we employ it as our analysis tool in this work.

\subsection{Our Contributions}
This work provides an analytical study on the performance in the link layer in terms of effective capacity at  FBL with imperfect CSI over multiple parallel fading channels.  Instantaneous CSI is assumed unavailable at the transmitter so that the transmitter must send pilot sequence to the receiver for channel estimation, then the receiver feeds back the estimated CSI.
The main contributions and findings are as follows:
\begin{itemize}
\item {{Effective Capacity in Exponential Integral Expression and System Parameters' Impact}: We first derive an expression in the form of exponential integral for the effective capacity of the parallel channels following Rayleigh fading at FBL with imperfect CSI. This expression facilitates the evaluation the effective capacity with respect to key system parameters, including the pilot length, the decoding error probability, the transmit power  as well as the  sub-channel number. More specifically, it is proved that the effective capacity is concave with respect to the pilot length $n_t$ and its inner term is also concave with respect to the decoding error probability $\varepsilon$ respectively. This indicates that there exists a unique optimal pilot length $n_t^*$ and unique optimal decoding error probability $\varepsilon^*$ that maximize the effective capacity respectively.
}

\item
 {{Closed-Form Lower Bound of Effective Capacity and Optimization Algorithm}: {With reasonable approximations, we  derive a closed-form lower bound of the effective capacity possessing the same properties with respect to the aforementioned parameters.
 Based on the closed-form lower bound, an alternating optimization-based algorithm is proposed to find the optimal pilot length and decoding error probability for maximizing the effective capacity at given transmit SNR and sub-channel number.}}

\item   {{Numerical Validation and Key Observations}: Numerical results validate that the lower bound of  effective capacity is quite tight over a wide range of system parameters.
   Furthermore, through the numerical results of the optimized effective capacity, it is shown that the optimal decoding error probability $\varepsilon^*$ decreases exponentially as the  sub-channel number $m$ or transmit SNR $\gamma_0$ (in dB) increases. The optimal pilot length $n_t^*$ keeps  constant as the  sub-channel number increases. When the transmit SNR $\gamma_0$  increases, the optimal pilot length $n_t^*$ decreases gradually and eventually keeps constant.}
\end{itemize}

The remainder of this paper is organized as follows:  Sec. II  presents the system model and  preliminaries of the FBL channel coding. Our main contributions are presented in Sec. III, identifying the impact of considered system parameters on the effective capacity and  then providing a closed-form lower bound of the effective capacity.  Numerical results are presented in Sec. IV. Finally, Sec. V concludes this paper.

\section{System Model and Effective Capacity}
 \begin{figure*}
\begin{centering}
\includegraphics[scale=0.62]{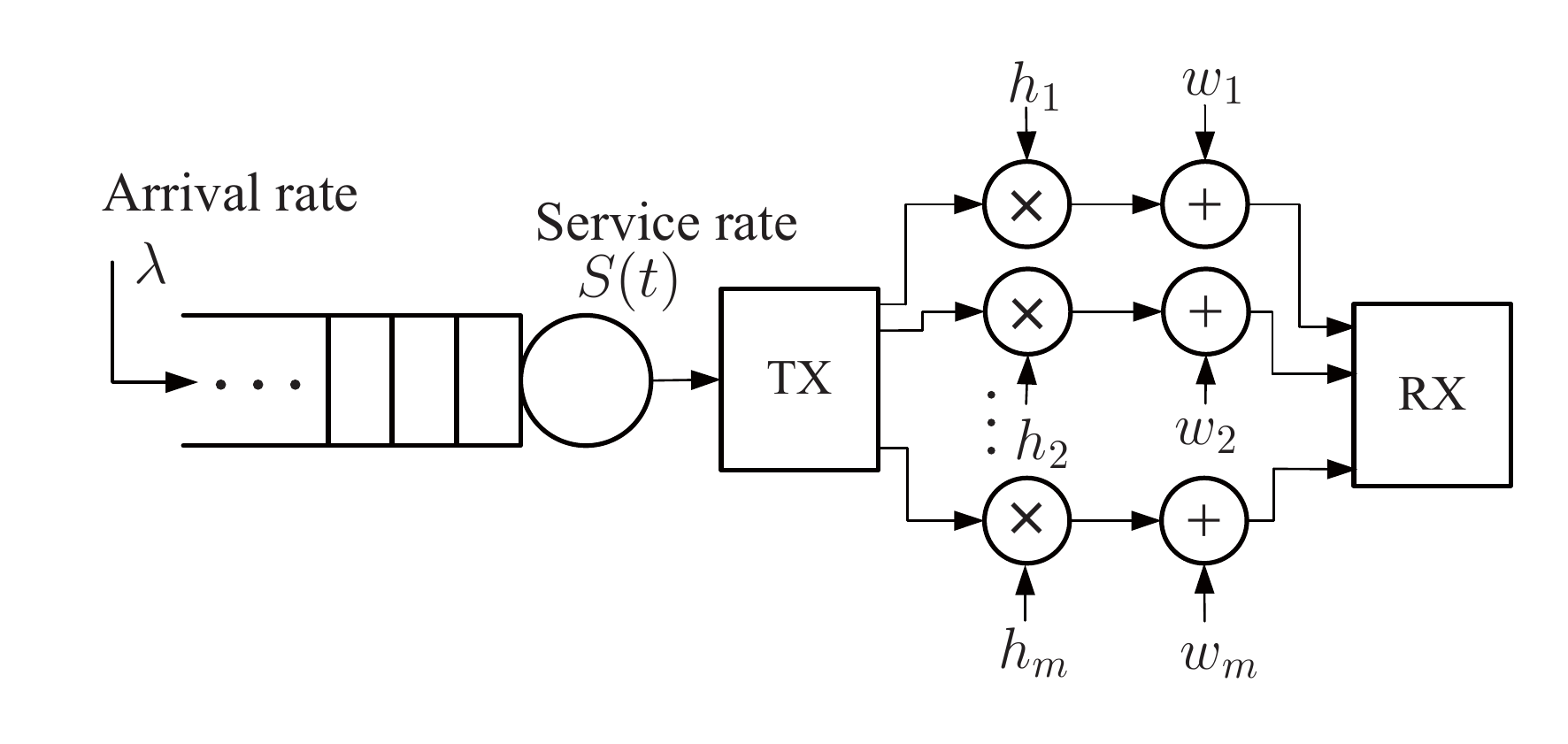}
 \caption{\small{System diagram for multiple parallel fading channels URLLC transmission.}}\label{fig:systemmodel}
\end{centering}
\end{figure*}

\begin{figure*}
\begin{centering}
\includegraphics[scale=0.60]{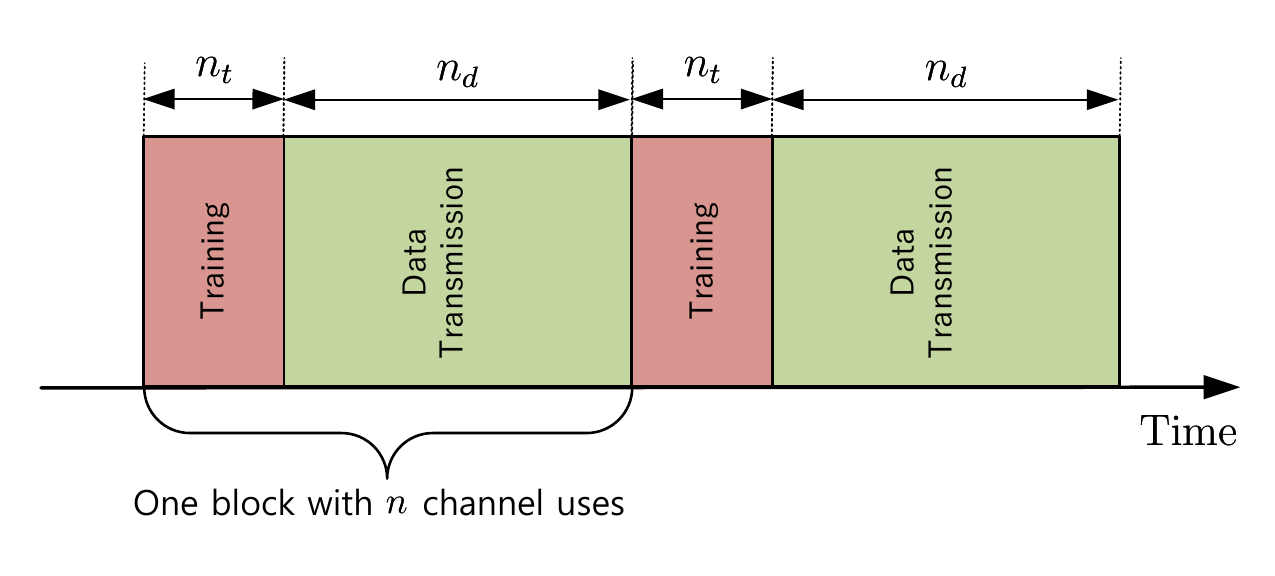}
 \caption{\small{Illustration of the block structure on each sub-channel: the transmitter first sends pilot to the receiver, then the receiver feeds back the fading coefficient to the transmitter (omitted in the figure), finally the transmitter sends data to the receiver according to the estimated CSI }}\label{fig:slotmodel}
\end{centering}
\vspace{-0.3cm}
\end{figure*}

 We consider data transmission from a source to a destination over multiple parallel fading channels using FBL coding as shown in Figure \ref{fig:systemmodel}. The information bits arrive at  a buffer of the transmitter and are kept in the buffer until they are successfully decoded at the receiver. The communication latency is determined by both the transmission delay in the physical layer and the queuing delay in the link layer. The performance metric is characterized by the effective capacity, i.e., the maximum constant arrival rate that can be supported under statistical delay constraints.

\subsection{Physical Layer Model}

Let $m$ denote the total number of parallel fading channels.  Each sub-channel is subject to independent and identically distributed (i.i.d) Rayleigh block  fading with additive white Gaussian noise. The fading coefficients, denoted as $\{{h_i}\}_{i=1}^m$,  remain constant within each block consists of $n \in \mathbb N $ channel uses and change independently from one block to another.  The channel gain is assumed to be normalized with $E[|h_i|^2] =1, \forall i$, and thus each channel coefficient follows the distribution $h_i \sim \mathcal{CN}(0,1)$.  The transmission block structure of each sub-channel  is  illustrated in  Figure~\ref{fig:slotmodel},  where the first $n_t=  \alpha n \in \mathbb N  $ channel uses of each block are used for training and the remaining $n_d= n - n_t$ channel uses are for data transmission. We ignore the CSI feedback part due to the fact that very few bits will be used to feed back the estimated fading coefficient to the transmitter. Furthermore, the feedback is assumed to be error-free and delay-free for simplicity of analysis.

In the training phase, the transmitter sends a pilot sequence of length $ n_t$ on each sub-channel for channel estimation.  The received pilot signal over the $i $-th sub-channel, for $i=1, ..., m$,  can be written as
\begin{equation}
\mathbf y^p_i=\sqrt{n_tp} h_i \mathbf q_i^H+\mathbf n_i,
\end{equation}
where $p$ is  the pilot transmit power at the transmitter, $\mathbf q_i$ is the pilot sequence with $ \mathbf  q_i^H \mathbf q_i=1$  and $\mathbf n_i \sim \mathcal {CN}(0,\sigma^2)$ is the additive white Gaussian noise.

Then the minimum mean square error (MMSE) estimate of the channel coefficient $h_i$ is given by \cite{5466522},\cite{1193803}
 \begin{equation}
\hat h_i=\frac{n_t \gamma_0}{1+n_t p} h_i + \frac{\sqrt{n_t \gamma_0}}{1+n_t \gamma_0} n_i
\end{equation}
where we define the transmit SNR as  $\gamma_0=\frac{p}{\sigma^2}$. Hence,  $\hat h_i$ is aslo Gaussian distributed as $\hat h_i \sim \mathcal {CN}(0,\frac{ n_t \gamma_0}{1+ n_t \gamma_0})$.

The relationship between the MMSE estimate of the channel coefficient $\hat h_i$ and its true value $h_i$ can be given as \cite{5466522}
  \begin{equation}
 h_i=\hat h_i + z_i,
 \end{equation}
 where the channel estimation error  $z_i\sim \mathcal {CN}(0,\frac{1}{1+ n_t\gamma_0})$ is independent of $\hat h_i$.
After the training phase, the receiver feeds back the estimated channel coefficients to the transmitter using an ideal feedback link.

In the data transmission phase, both the receiver and the transmitter treat the estimated CSI as the realistic CSI. Thus the estimation error $z_i$ will be regarded as noise.   The received signal on sub-channel $i$ is given by:
\begin{equation}
y_i=\sqrt{p} h_i s_i+ w_i=\sqrt{p} \hat h_i s_i + \underbrace {\sqrt{p} z_i s_i + w_i}_{v_i},
\end{equation}
where $s_i$ is the  transmitted signal with unit-power on the $i$th sub-channel,  $p$ is the data transmit power on each sub-channel  and  $w_i$ is the   complex Gaussian noise with zero mean and variance $\sigma^2$. Here we have assumed the training phase and data transmission phase have the same transmit power and the same noise power.  According  to the property of MMSE estimation, the estimation error and the estimate channel coefficient $\hat h_i$ are independent.  Thus the second noise term $\{v_i\}$ in (4) can be viewed as independent to the first signal term.

 The instantaneous  received SNR of the $i$th sub-channel can be represented as
\begin{equation}
\hat \gamma_i=  \frac{p|\hat h_i|^2}{\sigma_v^2}=\frac{p|\hat h_i|^2}{\sigma^2 + \frac{p}{1+ n_t \gamma_0}}
= {\frac{ n_t \gamma_0^2}{1+\gamma_0 + n_t \gamma_0}}    x_i=Gx_i,  \label{received-SNR}
\end{equation}
  where $| \hat h _i|^2= \frac{ n_t \gamma_0}{1+ n_t \gamma_0} x_i$, $x_i$ is an exponential distributed random variable with probability density function $f_{\rm PDF}(x_i)=e^{-x_i} $, and $G = {\frac{ n_t \gamma_0^2}{1+\gamma_0 + n_t \gamma_0}} $ can be viewed as the average received SNR taking channel estimation error into account.

{ We assume that channel coding is performed across the $m$ sub-channels, } since these $m$ parallel channels with estimated channel coefficients can be viewed as parallel AWGN channels with different noise variances. According to \cite{5205834}, the achievable coding rate in bits per channel use is given by 
 \begin{equation}
 \begin{aligned}
 R (\{\hat \gamma_i\}_1^m, \varepsilon)=  \frac{1}{m} \sum \limits_{i=1}^{m} \mathrm \log_2 (1+\hat \gamma_i) -\sqrt\frac {\mathcal V(\{\hat \gamma_i\}_1^m) }{ n_d m} Q^{-1}(\varepsilon) +o\left( \frac{\log_2n_d}{m n_d} \right),  \label{fblcoding}
 \end{aligned}
 \end{equation}
 where  $\varepsilon$ is the block decoding error probability, $Q^{-1}(\cdot)$ is  the inverse of the Gaussian Q function
 \begin{equation}
 Q(x)=\frac{1}{\sqrt{2\pi}}\int  _x^\infty e^{-\frac{t^2}{2}}dt,
 \end{equation}
   and  $\mathcal  V$ is  the channel dispersion
 \begin{equation}
 \mathcal V (\{\hat \gamma_i\}_1^m)=  \frac{1}{m}\sum\limits_{i=1}^{m}  \left( 1-\frac {1} {(1+ \hat \gamma_i)^2  } \right)\log_2 ^2 e.
 \end{equation}
 The achievable  rate expression (\ref{fblcoding}) is known as the normal approximation. For simplicity of analysis, we  neglect the third term  and   approximate the channel dispersion term  with an upper bound  
 \begin{equation}
\mathcal V(\{\hat \gamma_i\}_1^m)\approx \log_2^2 e. \label{dispersion}
 \end{equation}
Then, an approximate lower bound of the achievable rate at FBL can be obtained as:
\begin{equation}
 \begin{aligned}
 R (\{\hat \gamma_i\}_1^m, \varepsilon)=  \frac{1}{m} \sum \limits_{i=1}^{m} \mathrm \log_2 (1+\hat \gamma_i) -\frac {\log_2 e }{\sqrt{ n_d m}} Q^{-1}(\varepsilon).  \label{Normal-approx}
 \end{aligned}
 \end{equation}
 This approximation is shown to be accurate at high SNR ($\geq 3$dB) in \cite{Schiessl2015}, especially when multiple parallel channels are considered.
\subsection{Link Layer Model and Effective Capacity with Imperfect CSI at FBL}

We assume that the transmitter applies a simple automatic repeat request (ARQ) mechanism as in \cite{QiaoAug.2019}. Namely, when each transmission is over, the receiver can reliably detect the error transmission and then send  a negative acknowledgement requesting for retransmission in case of transmission error. This feedback link is also assumed to be error-free and delay-free. When error occurs, the service rate of a specific transmission block  is zero. Therefore, the service rate  at the $j$th transmission block (in bits per block) can be expressed as
\begin{equation}
\begin{aligned}
s(j)=
\begin {cases}
0, &\quad  \mathrm {probability} \quad \varepsilon  \\
m n_d  R( \{\hat \gamma_i\}_1^m  ,  \varepsilon),& \quad \mathrm {probability} \quad 1-\varepsilon
\end{cases}.
\end{aligned}
\end{equation}
The sequence $\{s(j), j=1, 2, ...\}$ is a discrete-time stationary and ergodic stochastic service process.

Effective capacity is the maximum constant arrival rate that a given service process can support with statistical QoS guarantee specified by the QoS exponent $\theta$ \cite{1210731}.
It can be calculated as \cite{1210731}
\begin{equation}
C_E(\theta)=-\lim_{t\to \infty} \frac{1}{\theta t} \ln \mathbb E \{e^{-\theta \mathbf S(1,t)} \}  \rm  \; bits/block, \label{EC}
\end{equation}
where $\mathbf S(1,t)=\sum \limits_{i=1}^{t} s(i)$  is the accumulated amount of service bits up to the $t$-th transmission block.

Given the effective capacity $C_E(\theta) =\mu$, the probability that the  transmission delay of the arrived information bits in slot  $t$, denoted as $D(t)$,  exceeds a target delay bound $D_{\rm max}$, i.e., delay violation probability, can be expressed approximately as \cite{1210731}
\begin{equation}
\mathbb P(D(t)>D_{\rm max})\approx \eta e^{-\theta \mu D_{\rm max}},  \label{dvp}
\end{equation}
where $\eta$ is the buffer non-empty probability and the QoS exponent $\theta$ indicates the decaying speed.  When the effective capacity and the delay bound are given, for larger $\theta$, the delay violation probability decays faster. Therefore, throughput performance with different statistical QoS guarantee (i.e., delay violation probability) can be characterized appropriately by the effective capacity with different decaying exponent $\theta$.

By definition (\ref{EC}), the effective capacity  of $m$ parallel channels at FBL can be given by
\begin{align}
C_{E}(\theta)& = -\lim \limits_{t \to \infty}\frac{1}{\theta t} \ln \mathbb E\left\{ e^{-\theta \mathbf  S(1, t)} \right\} \\
& = -\lim \limits_{t \to \infty}\frac{1}{\theta t} \ln (\mathbb E\{ e^{-\theta s(j)} \})^t  \label{service}\\
&= -\frac{1}{\theta } \ln \mathbb E\left \{ e^{-\theta s(j)} \right\} \\
& = -\frac{1}{\theta } \ln \mathbb E \left\{\varepsilon +(1-\varepsilon)e^{-\theta m n_d  R (\{\hat \gamma_i\}_1^m, \varepsilon) } \right\}, \label{EC-fbl}
\end{align}
where the expectation is with respect to $\hat \gamma=[\hat \gamma_1, \hat \gamma_2,...,\hat \gamma_m]$. Note that (\ref{service}) follows from the fact that the services process $s(j)$ changes independently from one block to another.
By substituting (\ref{Normal-approx})  into (\ref{EC-fbl}), effective capacity is obtained as shown in (\ref{18})-(\ref{EC-mc-fbl}).
\begin{align}
C_E(\theta)&= -\frac{1}{\theta } \ln \mathbb E \left\{ \varepsilon +(1-\varepsilon)e^{-\theta m n_d \left\{ \frac{1}{m} \sum \limits_{i=1}^{m} \mathrm \log_2(1+\hat\gamma_i) -\sqrt {\frac {\log_2^2e}{m n_d} }Q^{-1}(\varepsilon) \right\}  } \right\}  \label{18}\\
& =  -\frac{1}{\theta  } \ln   \left\{\varepsilon +(1-\varepsilon)e^{\theta\sqrt {mn_d} Q^{-1}(\varepsilon) \log_2 e}  \mathbb E \left[ \prod \limits_{i=1}^{m} (1+\hat \gamma_i)^{-\theta n_d \log_2e} \right]   \right\} \\
&=-\frac{1}{\theta  } \ln   \left\{\varepsilon +(1-\varepsilon)  \mathbb E \left[ e^{-\theta n_d\log_2 (1+\hat \gamma_i)+\theta n_d\sqrt {\frac{\log_2^2e}{mn_d}} Q^{-1}(\varepsilon) }\right] ^m   \right\}. \label{EC-mc-fbl}
\end{align}
\vspace{-0.7cm}
\newcounter{mytempeqncnt}

 {\subsection{{Key system parameters}}}

The main focus of this paper is to investigate the impact of the key system parameters, including the pilot length $n_t$, decoding error probability $\varepsilon$, transmit power $p$ and the sub-channel number $m$ on the effective capacity. They are elaborated as follows:

\textbf{Impact of pilot length  $n_t$}: when using a long training sequence of $n_t$ symbols, the channel estimation becomes more accurate, allowing transmissions with higher reliability but leaving fewer symbols ($n_d=n- n_{t}$) for the data transmission. While when using a short training sequence, the channel estimation becomes inaccurate while the channel uses for the data transmission increases. Thus, the parameter $n_t$ should be chosen carefully. This is particularly the case for FBL transmission since  the blocklength reduction of the data transmission part deteriorates the communication performance more rapidly due to the second order penalty term in (\ref{fblcoding}).

\textbf{Impact of the decoding error probability $\varepsilon$}: when adopting a larger decoding error probability, the transmission rate is larger while more retransmissions will be needed; when adopting a smaller decoding error probability, the channel coding rate decreases while the number of retransmission will be small. Finding the optimal decoding error probability that maximizes the system throughput is quite important.

\textbf{Impact of the transmit power $p$ (i.e., transmit SNR)}:  The throughput increases as the transmit power increases by no means. But how does the throughput increase, especially in the FBL regime is still under investigated.  Answering this question may help us design more practical power control scheme for URLLC systems.

\textbf{Impact of the number of sub-channels $m$}: Similar to the transmit power, how does the throughput increase as the number of sub-channels increases? This impact can provide valuable insights for subcarrier allocation in the OFDM system.

\section{Performance Analysis}
In this section,  we first show the impact of the pilot length and the decoding error probability on the effective capacity and identify the optimal tradeoffs. Then the impact of  the transmit power as well as the  sub-channel number on  the effective capacity  are addressed.   We then derive a closed-form lower bound of the effective capacity. Finally, we propose an alternating optimization based algorithm to maximize the effective capacity through optimizing the pilot length and the decoding error probability iteratively.
\subsection{ Impact of the  Considered Parameters on Effective Capacity}
 In this subsection, we will investigate the impact of  the pilot length $n_t$, transmit power $p$, sub-channel number $m$ and decoding error probability $\varepsilon$ on the effective capacity.
 Firstly, let us consider the impact of the decoding error probability $\varepsilon$. Denote the inner function  in (18) with respect to $\varepsilon$ as shown in (\ref{innerfunc}).
 \begin{align}
T(\varepsilon)= \mathbb E \left\{ \varepsilon +(1-\varepsilon)e^{-\theta m n_d \left\{ \frac{1}{m} \sum \limits_{i=1}^{m} \mathrm \log_2(1+\hat\gamma_i) -\sqrt {\frac {\log_2^2e}{m n_d} }Q^{-1}(\varepsilon) \right\}  } \right\}. \label{innerfunc}
\end{align}
 \newtheorem{mythm}{Theorem}
 \newtheorem{mylem}{Lemma}
 \newtheorem{proof}{Proof}
 \begin{mythm}
  For given $\theta$, $m$, $n_t$ and $p$, $T(\varepsilon)$ is strictly convex in $\varepsilon$. \label{Theorem1}
\end{mythm}

\begin{proof}
 This theorem follows directly Theorem 1 in \cite{QiaoAug.2019}, with the only difference that our channel coding is performed over multiple sub-channels and while it is over multiple coherence blocks in \cite{QiaoAug.2019}.
\end{proof}
 \newtheorem{mar}{Remark}

Theorem \ref{Theorem1} indicates that there exists an optimal decoding error probability that maximizes the effective capacity when the other parameters are given. This result reveals the optimal tradeoff between the decoding error probability (retransmissions) and the channel coding rate in terms of effective capacity.

Then we will consider the impact of the pilot length. Note that we shall remove the integer constraint of $n_t$  in the following analysis.
 \begin{mythm}
 For given $\theta$, $m$, $\varepsilon$ and $ \gamma_0$,  when $\alpha=\frac{n_t}{n} \in (0, 0.2)$, the effective capacity $C_E(\theta)$ is concave in $n_t$. \label{Thm2}
\end{mythm}
\begin{proof}
Please see appendix \ref{Theorem2-proof}.
\end{proof}

Theorem \ref{Thm2} indicates that if the ratio of the pilot length over the entire transmission blocklength is upper bounded by a certain value, there exists a unique pilot length $n_t^*$ that maximizes the effective capacity $C_E(\theta)$. This result identifies the optimal tradeoff  between the pilot length and the payload length in terms of the effective capacity.

 \begin{mythm}
Assume that the received SNR $\hat \gamma_i> -3$dB, $\forall i \in \{1,2,...,m\}$, 
{for given $\theta$, $m$, $\varepsilon$ and $n_t$, the effective capacity $C_E(\theta)$ is concave and monotonically increasing with respect to $p$.}
\end{mythm}
\begin{proof}
Please see appendix \ref{Theorem3}.
\end{proof}

 This result indicates that under the given minimum received SNR assumption, as the transmit power increases, the effective capacity increases but in a diminishing manner. When the transmit power is large enough,   the inner term of (\ref{innerfunc}) is bounded by $\varepsilon$, hence the effective capacity is bounded by $-\frac{1}{\theta}\ln \varepsilon$. In this case and the decoding error probability becomes the main factor which affects the effective capacity.

 \begin{mythm}
 For given $\theta$, $p$, $n_t$ and $\varepsilon$, the effective capacity $C_E(\theta)$ is monotonically increasing with respect to $m$.
\end{mythm}
\begin{proof}
Please see appendix \ref{Theorem4}.
 \end{proof}

 {Similarly, there is also an obvious upper bound of the effective capacity $-\frac{1}{\theta} \ln \varepsilon$. This result indicates, as the sub-channel number increases, the effective capacity also increases. When the  sub-channel number is large enough, the effective capacity remains constant and the decoding error probability becomes the main factor which affects the effective capacity.}

\subsection{Closed-form Approximation of the Effective Capacity}
{The aforementioned analysis indicates the existence of the impact of the parameters  on the effective capacity. How to obtain the optimal parameters is a non-trivial   problem,  because the effective capacity in (\ref{EC-mc-fbl}) is still in an integral form.} In this subsection, we will provide a closed-form lower bound  expression of the effective capacity  and verify the corresponding properties.

The effective capacity can be transformed as shown in (\ref{EC-def})-(\ref{EC-expint}),
\begin{align}
C_E(\theta)&= -\frac{1}{\theta } \ln \mathbb E \left\{ \varepsilon +(1-\varepsilon)e^{-\theta m n_d \left\{ \frac{1}{m} \sum \limits_{i=1}^{m} \mathrm \log_2(1+\hat\gamma_i) -\sqrt {\frac {\log_2^2e}{m n_d} }Q^{-1}(\varepsilon) \right\}  } \right\}   \label{EC-def} \\
&= -\frac{1}{\theta } \ln  \left\{\varepsilon + (1-\varepsilon)e^{\theta'\sqrt {mn_d} Q^{-1}(\varepsilon) } \left\{ \mathbb  E\left[ (1+ G x_i)^{-\theta' n_d } \right] \right\} ^{m}  \right\} \label{EC-indep} \\
&=  -\frac{1}{\theta }\ln  \left\{\varepsilon +(1-\varepsilon)e^{\theta'\sqrt {mn_d} Q^{-1}(\varepsilon)}  \left[ \int_0^\infty (1+ x_i)^{-\theta' n_d } e^{-x} dx_i \right] ^m\right\} \\
&= -\frac{1}{\theta } \ln  \left\{\varepsilon + (1-\varepsilon)e^ {\theta' \sqrt {mn_d} Q^{-1}(\varepsilon)}   \left[ \frac{1}{ G } e^{\frac{1}{ G}} E_{\theta' n_d }\left(\frac{1}{ G}\right)\right]^m  \right\}\label{EC-expint},
\end{align}
therein,  $\theta'=\theta \log_2e$  and $E_v(x)$ is the exponential integral given by
 \begin{equation}
 E_v(x)=\int_1^\infty e^{-xt}{t^{-v}}dt.
 \end{equation}
 From (\ref{EC-def}) to (\ref{EC-indep}), we employ the fact that all the sub-channels are i.i.d.
Then, by applying the upper bound of exponential integral in \cite{Chiccoli1990}
\begin{equation}
E_v(x) \leq \frac{e^{-x}}{v+x-1}, v> 1, \label{expint-lb}
\end{equation}
 a closed-form lower bound of the effective capacity is obtained in (\ref{EC-closed}).
\begin{equation}
\begin{aligned}
& \underline  C_E(\theta)=  -\frac{1}{\theta } \ln \left\{\varepsilon + (1-\varepsilon)  e^{\theta'\sqrt {mn_d} Q^{-1}(\varepsilon) } \left[ {(\theta'n_d-1) G+1} \right]^{-m} \right\}. \label{EC-closed}
\end{aligned}
\end{equation}
Note that one should have $\theta' n_d>1$ i.e.,  $ \theta n_d \log_2 e >1$ to make the upper bound (\ref{expint-lb}) valid. Furthermore,  $ \theta n_d \log_2 e >2$ can ensure the upper bound (\ref{expint-lb})'s tightness and in the following, we will mainly focus on this condition. This condition is reasonable because in the URLLC transmission, the QoS exponent $\theta$ is larger than $0.01$ and the blocklength is larger than 200. The closed-form lower bound (\ref{EC-closed}) provides an explicit relationship among the QoS exponent $\theta$, blocklength $n$, decoding error probability $\varepsilon$, average received SNR $ G$, and the sub-channel number $m$.

 It is obvious that the properties of the closed-form lower bound  $\underline  C_E(\theta)$ with respect to $\varepsilon$ and $m$ are the same as  $C_E(\theta)$.
 Next we will verify  the impact of the pilot length $n_t$ and the transmit power $p$ on the effective capacity lower bound $\underline  C_E(\theta)$.
  Firstly, fix $\varepsilon$  and remove it for simplicity. Furthermore, we employ the variable $n \alpha$ instead of the integer $n_t$ for convenience.  Here we assume $\alpha=\frac{n_t}{n}$ is constrained by  $\alpha \in (0,0.2)$. Then we denote $\Gamma(m, \gamma_0, \alpha)$ as shown in (\ref{EC-closed-e-origin})-(\ref{EC-closed-e}). From (\ref{EC-closed-e-origin})-(\ref{EC-closed-e-mid}), we employ the property of logarithmic function and from (\ref{EC-closed-e-mid})-(\ref{EC-closed-e}), we adopt the same approximation (\ref{sqrt-approx}) employed in the proof of Theorem 2 in the Appendix.
\begin{align}
\Gamma(m, \gamma_0, \alpha)&=-\frac{1}{\theta}\ln \left [e^{\theta' \sqrt{mn_d} Q^{-1}(\varepsilon) } \left( (\theta' n_d -1) G +1\right )    ^{-m}\right ] \label{EC-closed-e-origin}\\
&=-\sqrt{mn_d}Q^{-1}(\varepsilon)\log_2 e+\frac{m}{\theta}\ln((\theta'n_d-1) G+1)\label{EC-closed-e-mid}\\
&\approx -\sqrt{mn}Q^{-1}(\varepsilon)\log_2 e(1-\frac{\alpha}{2})+\frac{m}{\theta}\ln((\theta'n(1-\alpha)-1) G+1). \label{EC-closed-e}
\end{align}

\begin{mythm}
 Assume that $\theta n_d \log_2 e >2$, for given $\theta$, $m$, $\varepsilon$ and $p$,  when $\alpha \in (0, 0.2)$, the closed-form lower bound of the effective capacity $\underline  C_E(\theta)$ is concave over $\alpha$.
\end{mythm}
\begin{proof}
Please see Appendix \ref{proposition1}.
\end{proof}

It is easy to prove that the optimal pilot length $n_t^*$ that maximizes $\Gamma(m, \gamma_0, \alpha)$ also maximizes $\underline  C_E(\theta)$.
Thus, to obtain the optimal pilot length $n_t^*$,  we can directly calculate the root of the first order partial derivative of $\Gamma(m, \gamma_0, \alpha)$ in (\ref{EC-closed-e}) with respect to $\alpha$ and the solution can be found by using binary  search method. Note that $\alpha$ is a continuous  variable and is in the interval (0,0.2),  while $n_t$ is an integer, thus when we find the optimal $\alpha^*$, it necessary to compare the two adjacent integers to obtain the optimal $n_t^*$.
\begin{mythm}
 Assume that the average received SNR $G>-3$dB and $\theta n_d \log_2 e >2$, for given $\theta$, $m$, $\varepsilon$ and $n_t$,  the closed-form lower bound of the effective capacity $\underline  C_E(\theta)$  is concave and monotonically increasing with respect to $p$.
\end{mythm}
\begin{proof}
Please see Appendix \ref{proposition2}.
\end{proof}

This theorem indicates that the closed-form lower bound has the same  properties as the original function of effective capacity.
 \subsection{Joint Optimization of  Pilot Length and Decoding Error Probability }
In this subsection, our purpose is to maximize the effective capacity at given transmit power $p$ and the sub-channel number $m$ by jointly optimizing the pilot length $n_t$ and the decoding error probability $\varepsilon$. The problem is challenging since the effective capacity is not jointly concave in $\left(n_t, \varepsilon\right)$. However, the effective capacity is concave in both $n_t$ and $\varepsilon$ individually, here we adopt the the alternating optimization method, which can guarantee convergence. We first give the initial value of a decoding error probability. The optimal $n_t$ can be calculated through  the first order derivative of (\ref{EC-closed-e}) with respect to $\alpha$.  Binary search method can find the optimal pilot length $n_t=\lfloor  n \alpha \rfloor$  or $n_t=\lceil n \alpha \rceil$  efficiently. Then we take the derived $n_t$ as constant, the optimal decoding error probability can be updated through the first order derivative or directly obtained by leveraging ternary search method  of (\ref{EC-closed}) with respect to $\varepsilon$. Indeed, the expression of the first order derivative of (\ref{EC-closed}) is very complex, the ternary search method is employed  in this paper. Based on previous analysis, we can obtain the optimal parameters in each iteration. Thus the convergence of the proposed method can be guaranteed.  The overall procedure is outlined in Algorithm 1.

\algdef{SE}[DOWHILE]{Do}{doWhile}{\algorithmicdo}[1]{\algorithmicwhile\ #1}
\begin{algorithm}[h]
\caption{Alternating iterative method for effective capacity maximization }
\begin{algorithmic}[0]
\State { $\mathbf {Input}$: Number of sub-channels $m$, transmit SNR $\gamma_0$;}
\State { $\mathbf {Output}$: Optimal pilot length $n_t$, optimal decoding error probability $\varepsilon$;}
\State { Initialize $\varepsilon(0)=10^{-3}$, $i=1$. }
\Do
\State{Calculate the optimal pilot length  $n_t(i)$ that maximizes  (\ref{EC-closed-e}) at given $\varepsilon(0)$, then calculate the optimized effective capacity $C_{E_e}(i)$ by (\ref{EC-closed});}
\State{Calculate the optimal decoding error probability $\varepsilon(i)$ that maximizes (\ref{EC-closed}) at given pilot length $n_t(i)$, then calculate the optimized effective capacity $C_{E_{n_t}}(i)$ by (\ref{EC-closed});}
\State{$i=i+1$;}
\doWhile {$C_{E_{n_t}}(i)-C_{E_e}(i)>10^{-4}$ }\\
\Return{$\varepsilon(i), n_t(i)$}.
\label{code:recentEnd}
\end{algorithmic}
\end{algorithm}

\section{Numerical Results}
In this section,  we evaluate our derived  effective capacity and the closed-form lower bound under different conditions.
 Due to the accuracy of the channel dispersion approximation in (\ref{dispersion}), we mainly focus on the medium and high SNR scenarios. We set the noise power $\sigma^2=1$  for simplicity and thus the transmit SNR is actually the value of transmit power $p$.
 Note that the low latency communication is considered, the QoS requirement of each user is quite stringent, we thus set the QoS exponent as 0.01 and the number of channel uses in each block is 300 unless otherwise addressed. Hence the approximation in (\ref{expint-lb}) is tight under this condition $\theta n_d\log_2 e > 2$.
 To elaborate this further, let us assume the target delay bound is $D_{max} =5$ and the effective capacity is $\mu = 300$. Then, according to (\ref{dvp}), we have the delay violation probability upper bounded by $P(D>D_{max})=e^{\theta \mu  D_{max}} \approx 4.5\times 10^{-5} $, which is small enough to be considered as ultra-reliable.

To validate the analysis of effective capacity, we present three curves with the same parameters generated respectively by the following three methods: analytical results computed directly from the closed-form lower bound (28), Monte-Carlo simulation based on the definition in (22), and analytical results computed from the exponential integral form (25). The purpose is two-fold. The first is to validate the correctness of the theoretical derivation (25) using simulation based on (22). The detailed Monte-Carlo simulation shall be introduced in due course. The second is to verify the tightness of the lower bound (28) by comparing with (25) or (22).

 In \ref{4A}, we  present numerical results to validate the theoretical  analysis.
 Then in Sec. \ref{4B}, we present the optimized effective capacity versus the  sub-channel number and the transmit SNR based on our proposed alternating iterative method, where the corresponding  optimal pilot length and optimal decoding error probability are also  presented.
 \subsection{ Theoretical Analysis and Approximation Accuracy Validation} \label{4A}
\begin{figure}[t]
\centering
\subfigure[Effective capacity v.s. $n_t$ ]
{\label{Fig:EC-nt}  \includegraphics[width=0.65\textwidth]{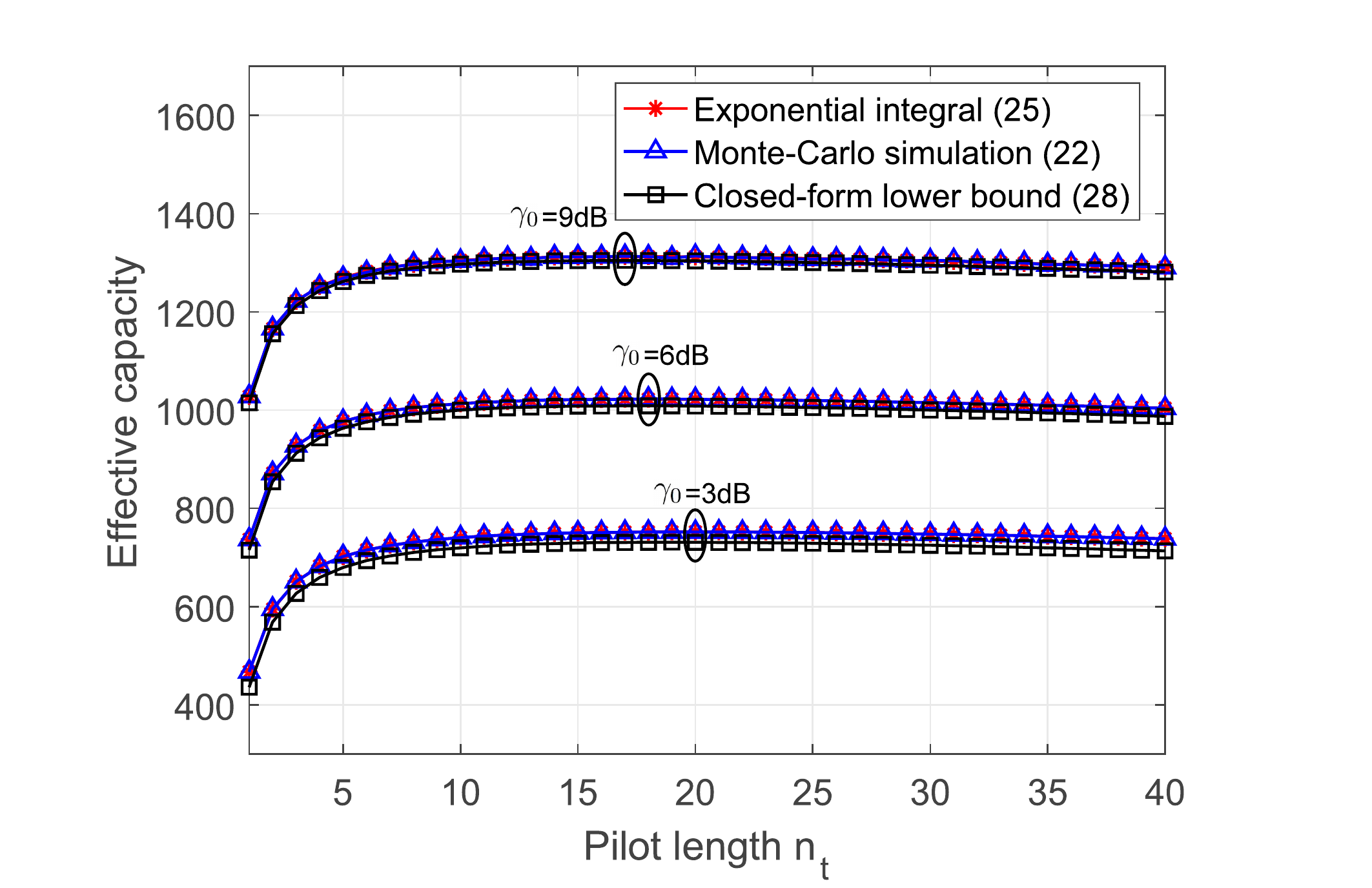}}

\subfigure[Effective capacity v.s. $\varepsilon$ ]
{\label{Fig:EC-e}
\includegraphics[width=0.65\textwidth]{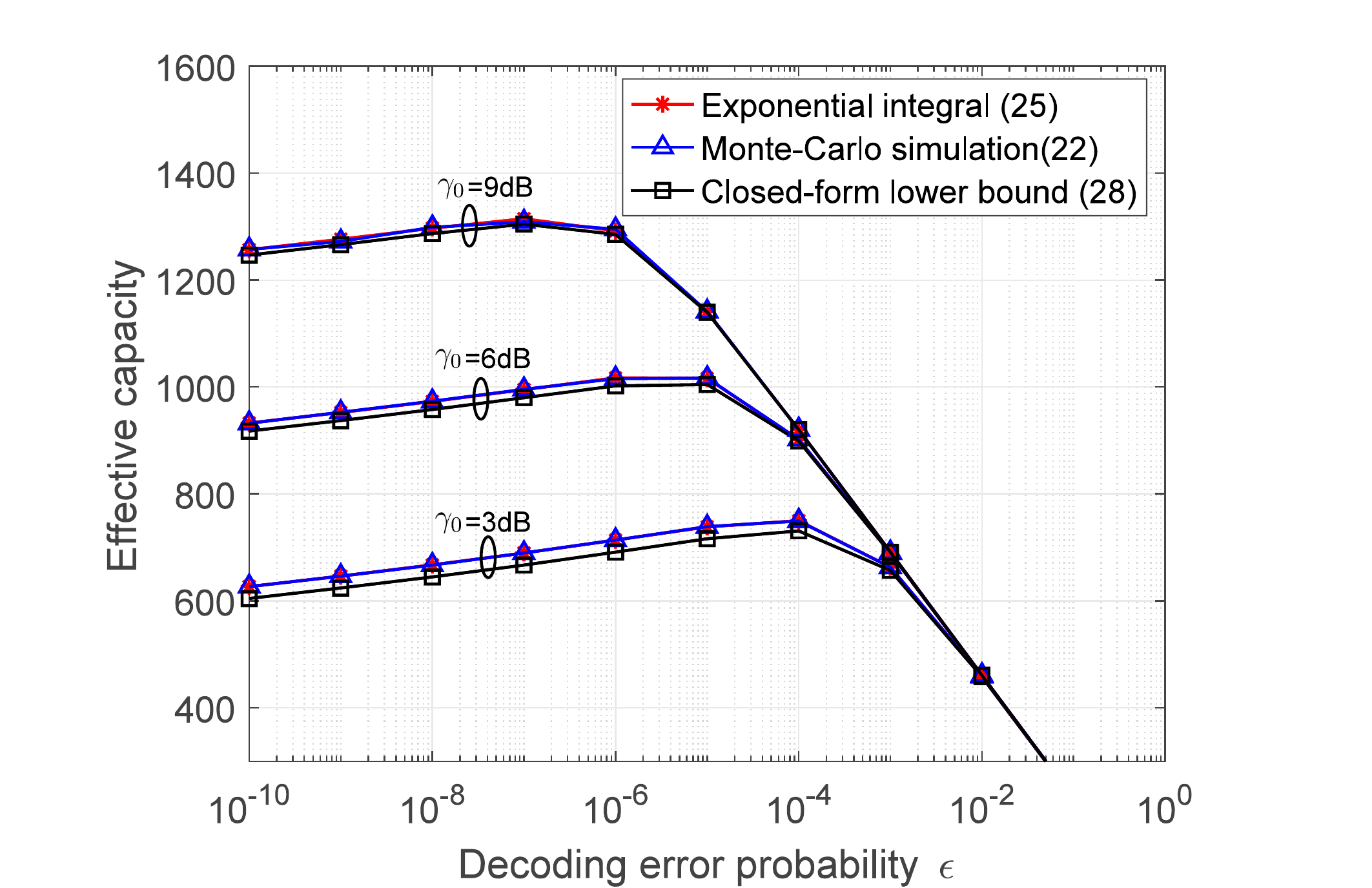}}

\caption{ Impact of  $n_t$, $\varepsilon$ on the effective capacity}
\label{Num1}
\vspace{-0.8cm}
\end{figure}

In Fig.~\ref{Num1}\subref{Fig:EC-nt}, we validate the analysis with respect to the pilot length $n_t \in \mathbb N$  and show the tightness of the closed-form lower bound with   $m=5$ sub-channels. The transmit SNR and corresponding decoding error probability of the three groups are $(3dB, 4.33\times 10^{-6})$, $(6dB, 7.92\times 10^{-5})$,  $(9dB, 2.02\times 10^{-7})$, we generate $10^{10}$ block realizations  to calculate the Monte-Carlo simulation results   based on the definition in (22). It is first seen that the simulation results based on (22) match exactly with the analytical  results based on the exponential integral form (25).
It is also seen that the effective capacity is concave with respect to  $n_t$ and the closed-form expression is  exactly a lower bound of the effective capacity. In addition, the gap between the closed-form lower bound and the effective capacity decreases as the transmit SNR increases.
The trends of the three curves with the same parameters are all concave and the optimal pilot length (points in the circles) that maximizes the effective capacity decreases as the transmit SNR increases.

In Fig.~\ref{Num1}\subref{Fig:EC-e}, we validate the theoretical analysis with respect to the decoding error probability $\varepsilon$ and show the tightness of the closed-form lower bound  for different transmit SNR $\gamma_0$. The number of sub-channels is $m=5$.
The transmit SNR and pilot length for the three groups are $(3dB, 19)$, $(6dB,18)$, $(9dB,17)$ respectively. We generate $10^{10}$ block realizations  to calculate the Monte-Carlo simulation results.
It is seen that the effective capacity is maximized by a specific decoding error probability as in \cite{QiaoAug.2019} and our proposed closed-form lower bound is exactly a lower bound of (\ref{EC-expint}).
The trends of the three curves with the same parameters stay consistent and it can be seen that the optimal decoding error probabilities of the three curves in the same group are almost the same, which also shows the strength of our proposed closed-form lower bound. Furthermore, we can observe that the optimal decoding error probability decreases as the transmit SNR increases when the other parameters stay constant.
The gap among the three groups are very small when $\varepsilon$ is large, while when $\varepsilon$ becomes smaller, the gap increases.  The gap between the closed-form lower bound and the other two curves is inherently due to the exponential integral upper bound  in (\ref{expint-lb}). With  higher transmit SNR, the upper bound is more tight and the gap becomes smaller.

\begin{figure}[t]
\centering
\subfigure[Effective capacity v.s. $\gamma_0$ ] {
\label{Fig:EC-p}
\includegraphics[width=0.65\textwidth]{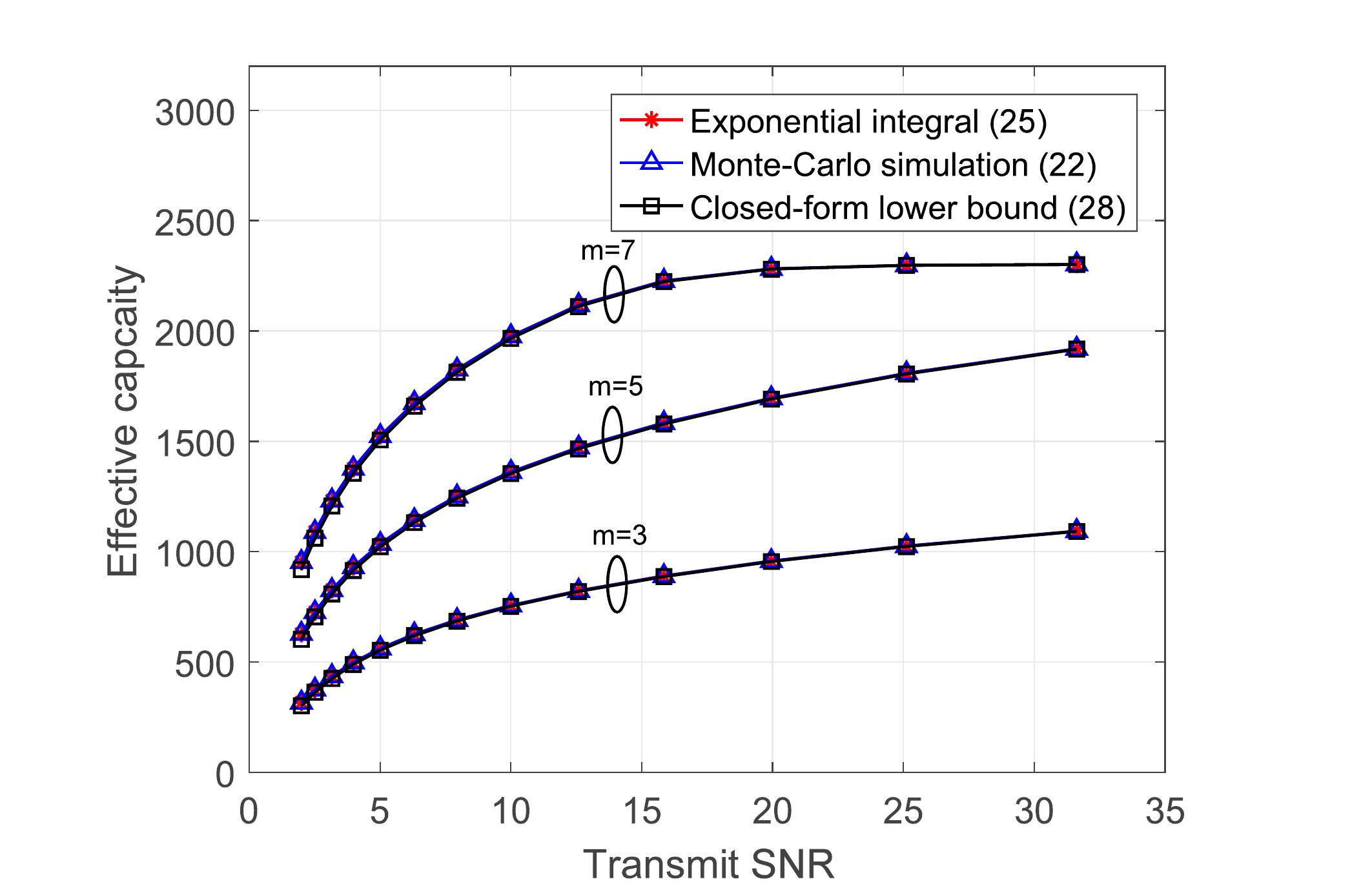} }

\subfigure[Effective capacity v.s. $m$ ] {
\label{Fig:EC-m}
\includegraphics[width=0.65\textwidth]{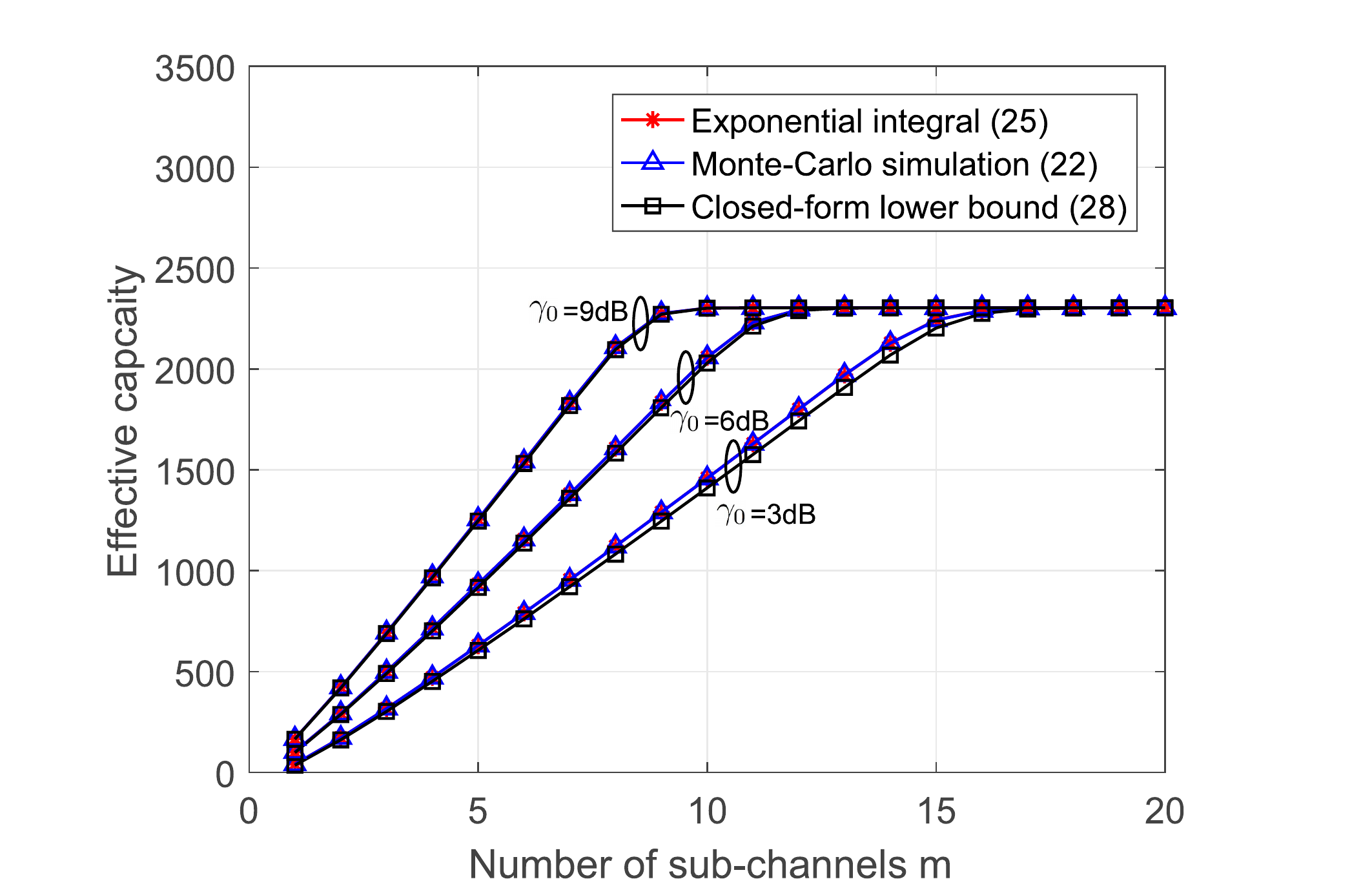} }

\caption{ Impact of $p$, $m$ on the effective capacity}
\label{Num2}
\vspace{-0.8cm}
\end{figure}

In Fig.\ref{Num2}\subref{Fig:EC-p}, we show the effective capacity and the closed-form lower bound as a function of the transmit SNR $\gamma_0$ (not dB for intuition) with  $\varepsilon=10^{-10}$, $ n_t=20 $ and set 3 groups with different values of $m$ for comparison. We generate  $ 10^{11}$ block realizations to evaluate the Mote-Carlo simulation of each point.  Firstly, it is observed that the effective capacity is concave and monotonically increasing with respect to $\gamma_0$, which validates our analysis. Moreover, when $m=7$, the effective capacity approaches the limit $-\frac{1}{\theta}\ln \varepsilon$ as $\gamma_0$ increases.
{With given decoding error probability, there exists a transmit SNR threshold, when the transmit SNR is higher than this threshold, the effective capacity increases very slowly. When the  sub-channel number  $m$ increases, the threshold decreases.} Secondly, the gap between the  effective capacity and closed-form lower bound decreases as the transmit SNR increases.

In Fig.~\ref{Num2}\subref{Fig:EC-m}, we show the effective capacity and the closed-form lower bound with respect to the  sub-channel number $m$.  Here  we assume
$n_t=20$ and show different transmit SNR values for comparison. We generate  $ 10^{11}$ block realizations to evaluate the Mote-Carlo simulation  of each point. It is seen that the effective capacity is monotonically increasing with respect to $m$ and then converges to $-\frac{1}{\theta}\ln \varepsilon$, which proves Theorem 4. {With given decoding error probability, there also exists a sub-channels threshold, when the  sub-channel number is larger than the threshold, the effective capacity increases very slowly. When the transmit SNR increases, the threshold decreases and  the gap between the effective capacity and the closed-form lower bound is smaller.}

From Fig.~\ref{Num2}\subref{Fig:EC-p}, it is seen the closed-form lower bound (\ref{EC-closed}) is very tight. While as shown in from Fig.~\ref{Num2}\subref{Fig:EC-m}, there is a small gap for the group  $\gamma_0=3dB$  when  $m\geq4$. This gap is due to the number of sub-channels on the exponential position amplify the exponential integral approximation error in (\ref{expint-lb}) and thus with larger $m$, the gap is bigger. The gap finally diminishes because all of the three curves within the same group gradually converge to $-\frac{1}{\theta}\ln \varepsilon$.

\subsection{Optimized Effective Capacity}\label{4B}
In this subsection, we present the optimized effective capacity as a function of the  sub-channel number and the transmit SNR. Then we show the corresponding variations of the optimal pilot length $n_t^*$ and the optimal decoding error probability $\varepsilon^*$.
\begin{figure}[t]
\centering
\subfigure[Optimized effective capacity v.s. $m$] {
\label{Fig:EC-opt-m}
\includegraphics[width=0.65\textwidth]{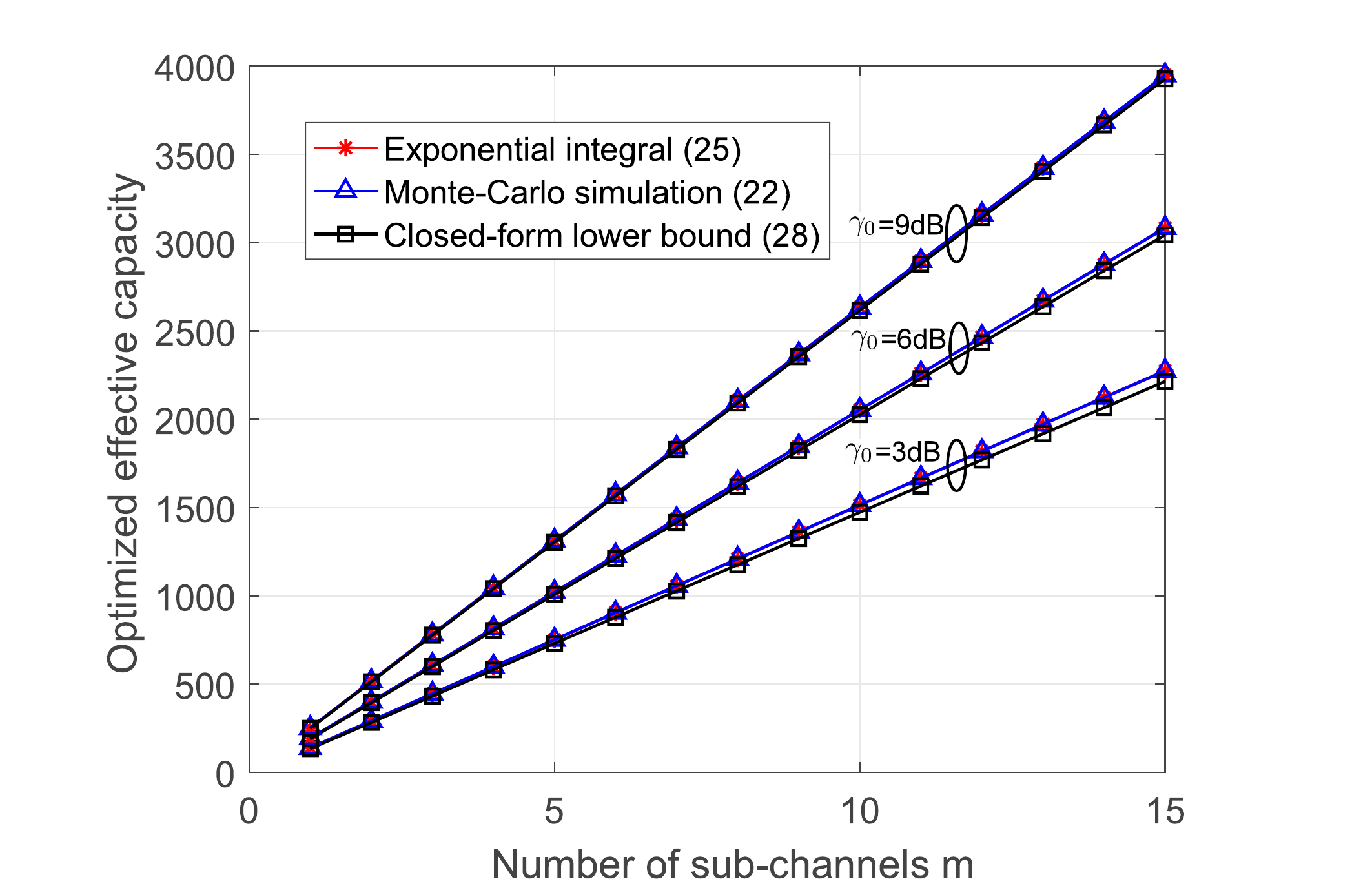}  }

\subfigure[Optimized effective capacity v.s. $p$ ] {
\label{Fig:EC-opt-p}
\includegraphics[width=0.65\textwidth]{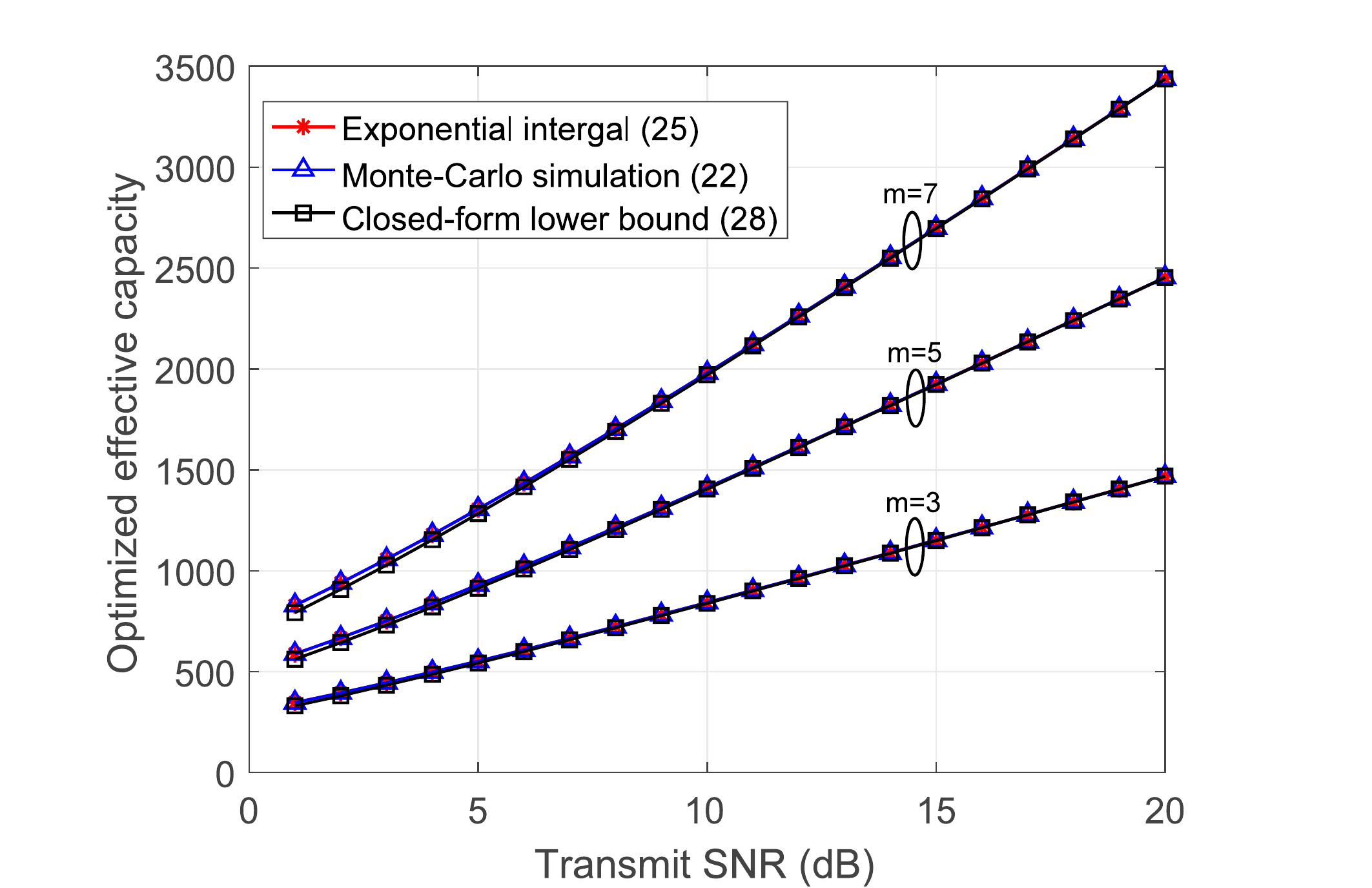}  }
\caption{ Optimized effective capacity v.s. $m$ and $p$}
\label{Num3}
\vspace{-0.8cm}
\end{figure}

In Fig.~\ref{Num3}\subref{Fig:EC-opt-m}, we show the  optimized effective capacity and the closed-form lower bound versus the number of sub-channels with different transmit SNR $\gamma_0$.
We can see that the optimized effective capacity is monotonically increasing alomost linearly with respect to $m$. For higher transmit SNR $\gamma_0$, the slope of the optimized effective capacity is larger.

We show the optimized effective capacity as function of the transmit SNR $\gamma_0$ (in dB) with different $m$ in Fig.~\ref{Num3}\subref{Fig:EC-opt-p}.
The effective capacity is concave  and  monotonically increasing with respect to the transmit SNR $\gamma_0$ and the closed-form lower bound is very close to the effective capacity.  Furthermore, for larger $m$,  the slope of the optimized effective capacity is also larger.

\begin{figure}[t]
\centering
\subfigure[Optimal pilot length $n_t$ v.s. $m$ ] {
\label{Fig:EC-opt-nt}
\includegraphics[width=0.65\textwidth]{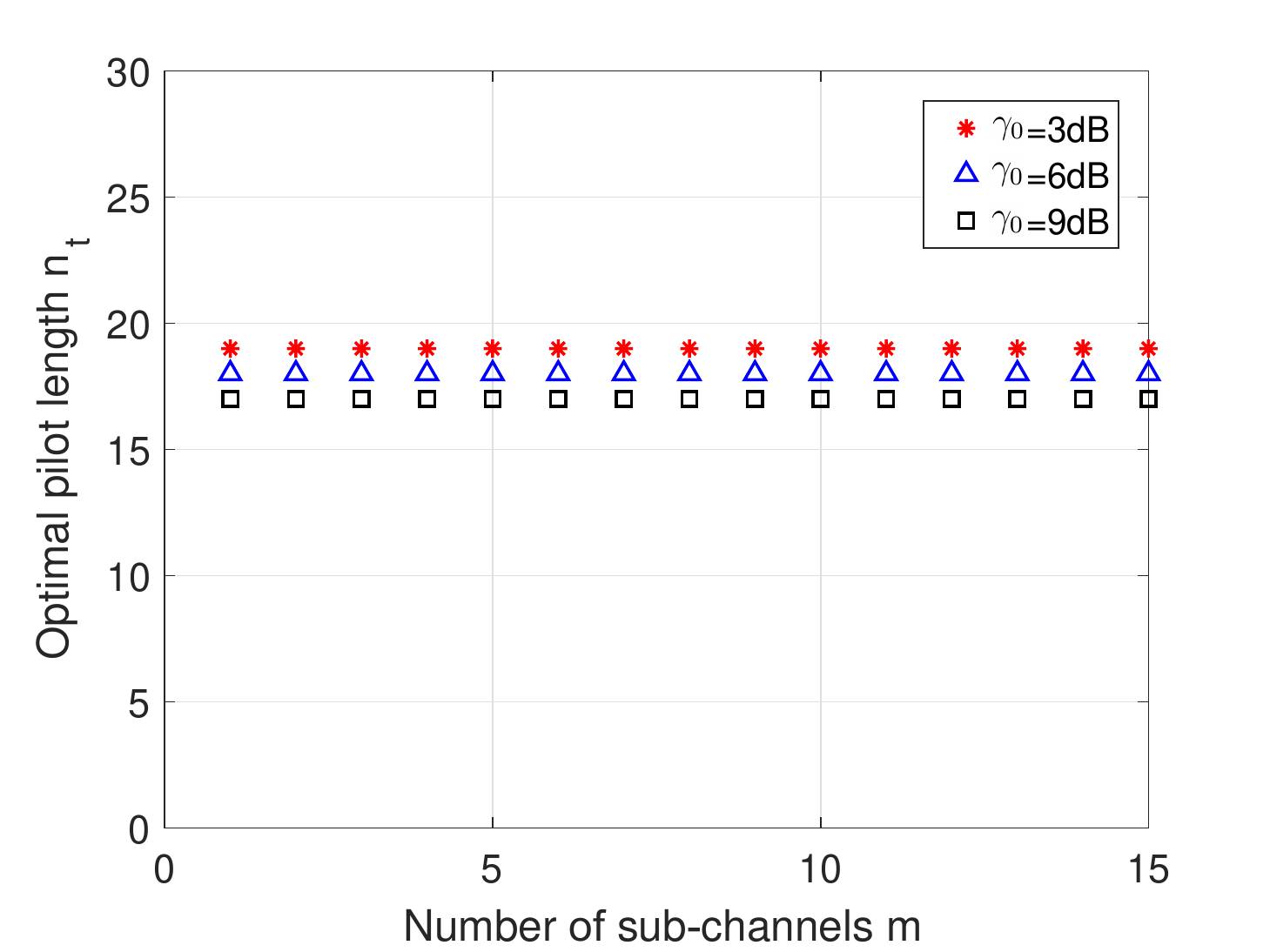}  }
\subfigure[Optimal pilot length $n_t$ v.s. $\gamma_0$] {
\label{Fig:EC-opt-nt-2}
\includegraphics[width=0.65\textwidth]{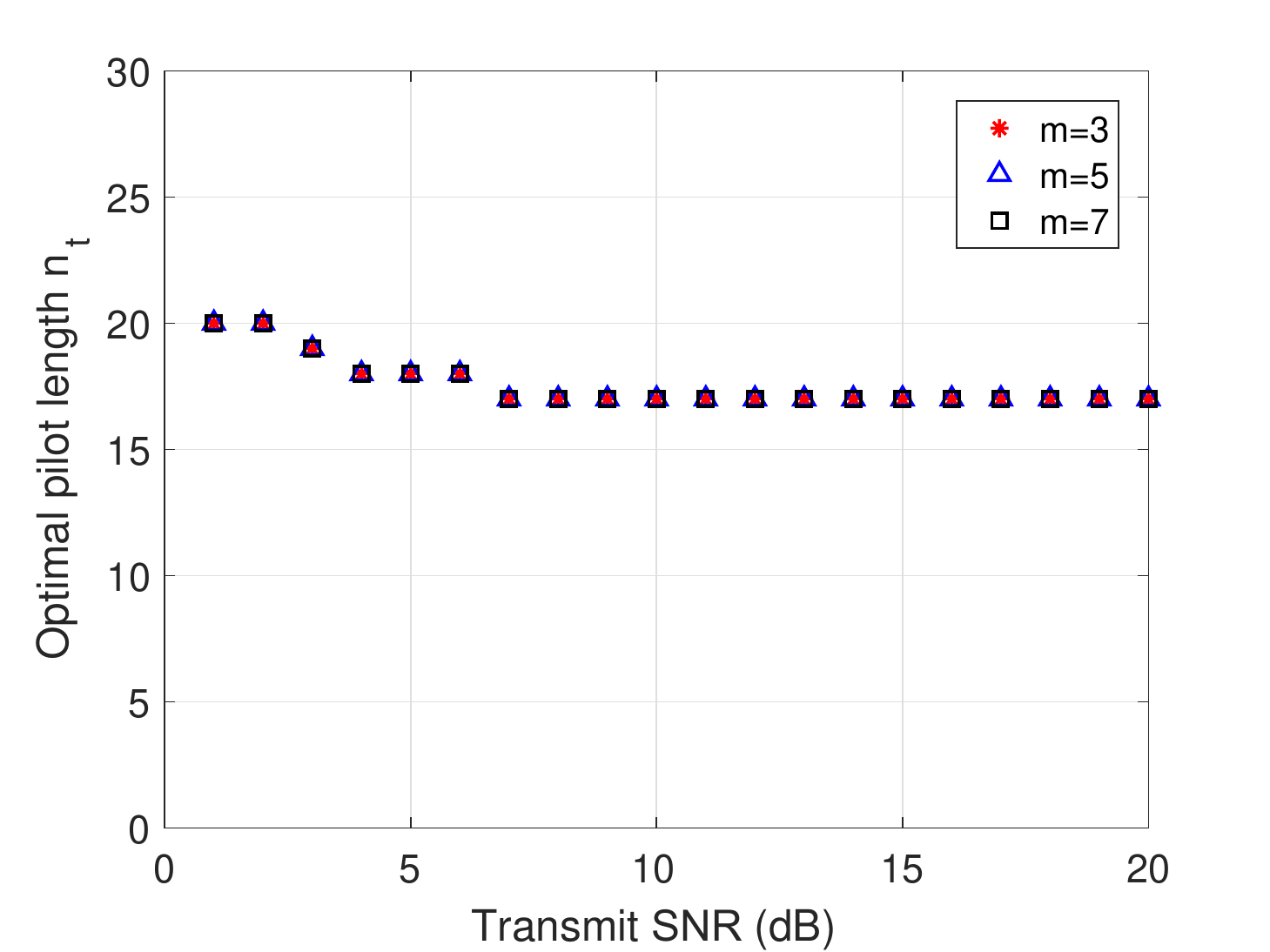}  }
\caption{ Optimal pilot length of the optimized effective capacity}
\label{Num4}
\vspace{-0.8cm}
\end{figure}

Fig.~\ref{Num4}\subref{Fig:EC-opt-nt} shows the optimal pilot length of the optimized effective capacity as a function of $m$ for different transmit SNR $\gamma_0$.  The optimal pilot length of given transmit SNR doesn't change as the number of sub-channels increases. While the optimal pilot lengths with different transmit SNR are different. Specifically,  higher SNR corresponds to shorter pilot length.

Fig.~\ref{Num4}\subref{Fig:EC-opt-nt-2} shows the optimal pilot length of the optimized effective capacity as a function of $\gamma_0$ with different  $m$.  It is seen that the optimal pilot lengths with different $m$ are exactly the same and gradually decrease  and finally converge to a specific value as the transmit SNR  increases in our simulated setting.

\begin{figure}[t]
\centering
\subfigure[Optimal decoding error probability $\varepsilon$ v.s. $m$ ] {
\label{Fig:EC-opt-e}
\includegraphics[width=0.65\textwidth]{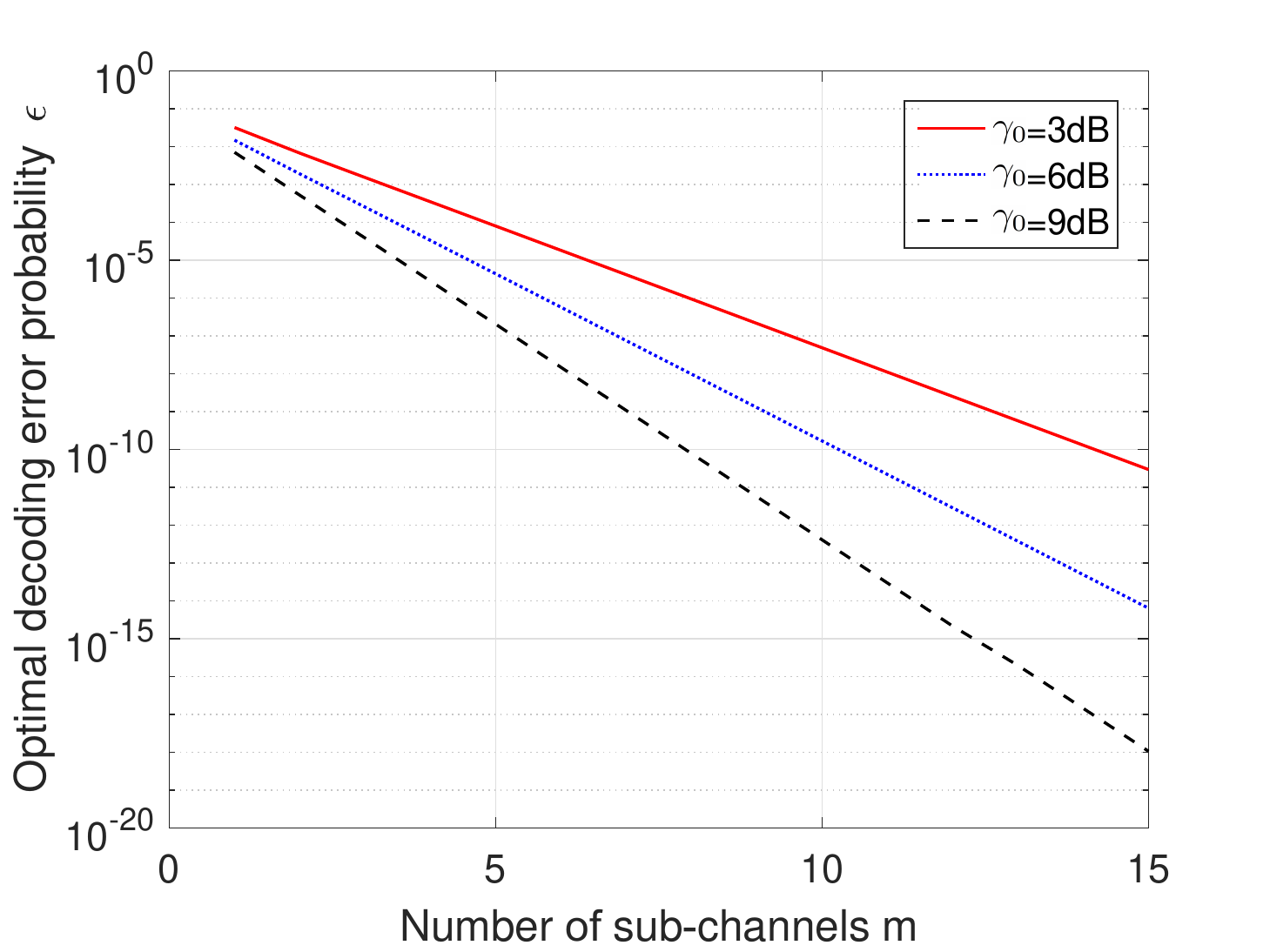}  }

\subfigure[Optimal decoding error probability $\varepsilon$ v.s. $\gamma_0$] {
\label{Fig:EC-opt-e-2}
\includegraphics[width=0.65\textwidth]{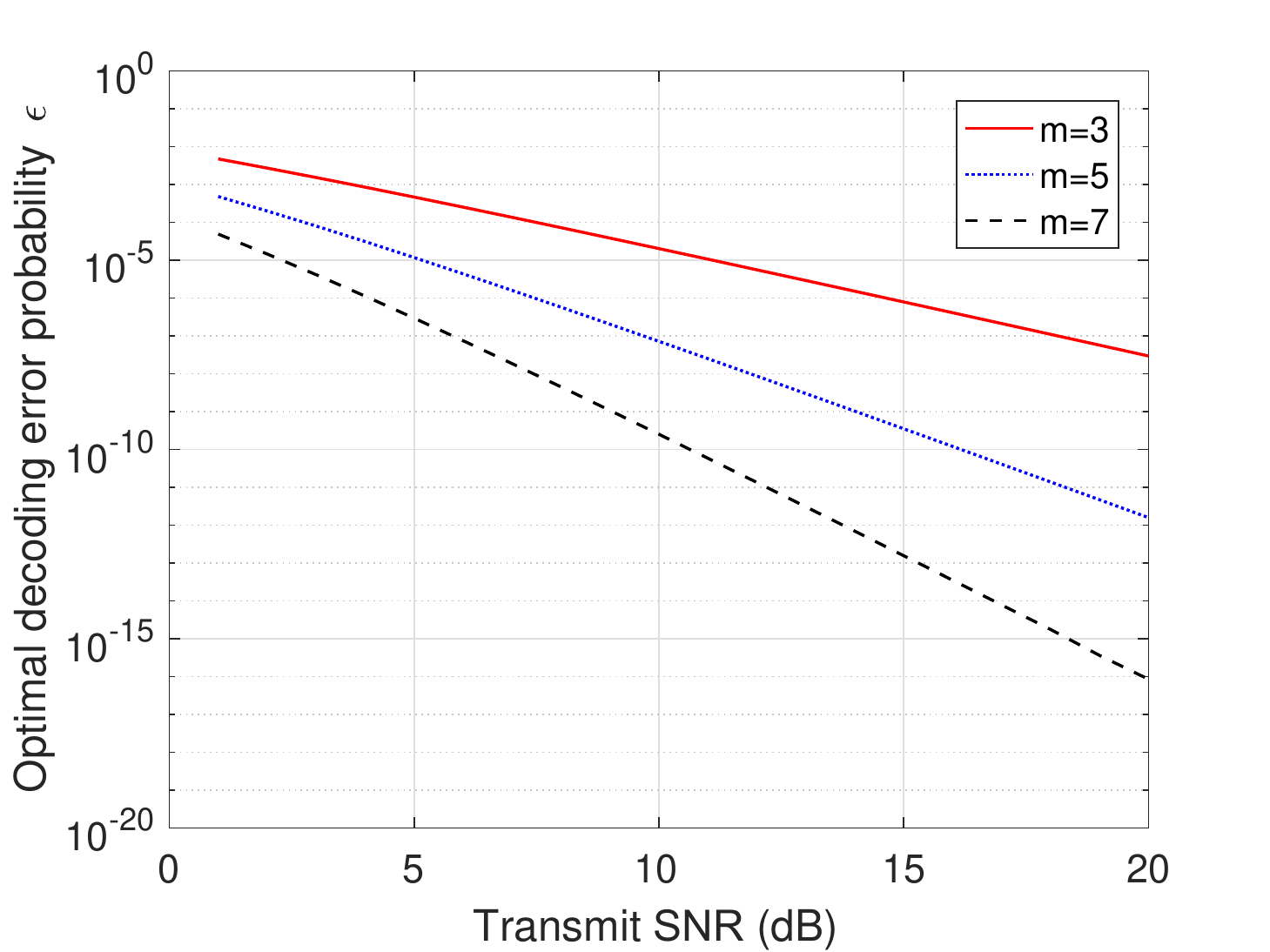} }
\caption{ Optimal decoding error probability of the optimized effective capacity }
\label{Num5}
\vspace{-0.8cm}
\end{figure}

Fig.~\ref{Num5}\subref{Fig:EC-opt-e}  depicts the corresponding optimal values of the decoding error probability as a function of sub-channel number $m$. The optimal decoding error probability decreases exponentially as the sub-channel number increases. For higher transmit SNR  $\gamma_0$, the optimal decoding error probability decays faster.

Fig.~\ref{Num5}\subref{Fig:EC-opt-e-2}  depicts the corresponding optimal values of the decoding error probability as a function of the transmit SNR $\gamma_0$. The optimal decoding error probability decreases exponentially as the transmit SNR (in dB) increases. For larger $m$, the optimal decoding error probability decays faster.

As a final remark, the main insights from the above numerical study are as follows: As $m$ or $\gamma_0$ increases, by optimizing  $n_t$ and $\varepsilon$ iteratively, the optimized effective capacity can increase almost linearly. This is due to the optimal decoding error probability $\varepsilon^*$ is decreasing exponentially, the transmission rate is not constrained by  $\varepsilon^*$, i.e., the tradeoff between decoding error probability and coding rate is met. On the other hand, by optimizing the pilot length $n_t$, the transmission and training tradeoff can also be met. The optimized pilot length keeps constant  as the number of sub-channels changes and decreases gradually as the transmit power $\gamma_0$ increases. This result coincides with intuition, i.e., the channel estimation is related to the channel statistical characteristics  instead of the sub-channel number, and with higher transmit SNR, the shorter pilot length is required.
So in the realistic system, the effective capacity can be calculated rapidly under certain conditions, e.g., with given transmit SNR, if we obtain two points of the curves, we can estimate the effective capacity with more or less sub-channels and the corresponding optimized decoding error probability can also be updated easily utilizing a linear approximation.

\section{Conclusion}
In this paper, we studied the throughput  performance under statistical QoS requirement in terms  of  effective capacity over parallel fading channels at FBL with imperfect CSI.
Firstly, we analyzed the impact of the considered parameters  on the effective capacity.
Then we derived a closed-form lower bound by adopting reasonable approximation of the exponential integral and verified the aforementioned properties.
Furthermore, we propose an alternating optimization-based algorithm to maximize the effective capacity by optimizing decoding error probability and the pilot length iteratively with given sub-channels and transmit SNR based on our proposed closed-form lower bound.
Numerical results validated our analysis. Then we showed optimized effective capacity and corresponding optimal pilot length as well as optimal decoding error probability. As a results, this work can provide some insights to guide the design of key system parameters in practical URLLC systems.


\begin{appendices}
\section{Proof of Theorem 2} \label{Theorem2-proof}

Firstly, we fix $\varepsilon$ and define the following function  with respect to $m$, $\gamma_0$ and $n_t$ in (\ref{EC-simp}).
\begin{align}
L(m, \gamma_0, n_t)&=-\frac{1}{\theta  } \ln   \left\{ \mathbb E \left[ e^{-\theta n_d\log_2 (1+\hat \gamma_i)+\theta n_d \sqrt {\frac{\log_2^2e}{mn_d}} Q^{-1}(\varepsilon) }\right] ^m   \right \}.  \label{EC-simp}
\end{align}
Then, we transform the expression in (\ref{EC-simp}) as shown in (\ref{funcL}).
\begin{align}
L(m, p, n_t)&=-{\sqrt {mn_d} Q^{-1}(\varepsilon)\log_2 e}-\frac{m}{\theta} \ln \left\{\mathbb E \left [ e^{-\theta n_d \log_2 (1+\hat \gamma_i)}\right] \right\} \label{simp-1}\\
&\approx -{ \sqrt {mn} Q^{-1}(\varepsilon)\log_2 e (1-\frac{n_t}{2n})}-\frac{m}{\theta} \ln \left\{\mathbb E \left [ e^{-\theta n_d \log_2 (1+\hat \gamma_i)}\right]  \right\},
\label{funcL}
\end{align}

From (\ref{simp-1}) to (\ref{funcL}),  the following approximation is  employed
\begin{align}
 \sqrt{1-\frac{n_t}{n}}\approx 1-\frac{n_t}{2n},\label{sqrt-approx}
\end{align}
where this approximation is accurate enough when $\frac{n_t}{n}\in (0,0.2)$. Due to the linearity  of the first term, it doesn't alter the convexity of $L(n_t)$, then the focus is on the latter function. Note that the latter term is quite complicated and the pilot length $n_t$ is implicitly included in $n_d$ and $\hat \gamma_i$, we will show the convexity of $e^{-\theta n_d \left(\log_2 (1+\hat \gamma_i)\right)}$ and then show the convexity of the overall latter term through compound function with respect to $n_t$.  Denote
\begin{align}
\Psi(p, n_t)=\log_2 (1+\hat \gamma_i).
\end{align}

 The first order and second order partial derivatives of $\Psi(p, n_t)$ with respect to $n_t$ are derived as follows:
 \begin{equation}
 \frac{\partial \Psi}{\partial n_t}=\frac{\partial \Psi}{\partial \hat \gamma_i}\frac{\partial \hat \gamma_i}{\partial n_t}=\frac{\log_2 e}{1+\hat \gamma_i}\frac{\gamma_0^2(\gamma_0+1)x_i}{(1+\gamma_0 +\gamma_0 n_t)^2}\geq 0,
\end{equation}
\begin{equation}
\begin{aligned}
\frac{\partial^2 \Psi}{\partial n_t^2}=&\frac{\partial ^2\Psi}{\partial \hat \gamma^2}\left(\frac{\partial \hat \gamma}{\partial n_t}\right)^2+\frac{\partial \Psi}{\partial \hat \gamma}\frac{\partial^2 \hat \gamma}{\partial n_t^2}\\
=&- \frac{\log_2 e}{(1+\hat \gamma_i)^2}  \left(\frac{\gamma_0^2(\gamma_0+1)x_i}{(1+\gamma_0 +n_t \gamma_0 )^2}\right) ^2  \\
&-\frac{\log_2 e}{1+\hat \gamma_i}\frac{2\gamma_0^3(\gamma_0+1)x_i}{(1+\gamma_0 + n_t \gamma_0 )^3}<0.
\end{aligned}
\end{equation}
Therefore, $\Psi(n_t)$ is concave over $n_t$.
 Then the function $\Phi(n_t)=(n-n_t)\Psi(n_t)$ is also concave with respect to $n_t$. Because

 \begin{align}
\frac{\partial ^2 \Phi}{\partial n_t^2}=-\frac{\partial \Psi}{\partial n_t}-\frac{\partial \Psi}{\partial n_t}+(n-n_t)\frac{\partial ^2 \Psi}{ \partial n_t^2}<0.
\end{align}
The function $-\Phi(n_t)$ is convex over $n_t$,
and thus, the function $e^{-\theta \Phi(n_t)}$ is also convex respect to $n_t$.
To show the convexity of $\ln \mathbb E \left[ e^{-\theta \Phi(n_t)}\right] $, Lemma 1 is proposed.

\begin{mylem}
Denote the function
\begin{align}
h(x)=\ln \left(\sum \limits_{i=1}^n e^{a_i x} \right),
\end{align}
$h(x)$ is convex with respect to $x$.
\end{mylem}
\begin{proof}
The first order and second order derivatives of $h(x)$ are:
\begin{align}
h'(x)= \frac{\sum \limits_{i=1}^n a_i e^{a_i x}}{\sum \limits_{i=1}^n e^{a_i x}} \geq 0,
\end{align}
\vspace{-0.6cm}
\begin{align}
h''(x)=& \frac{\sum \limits_{i=1}^n \sum \limits_{j=1}^n a_i^2e^{(a_i +a_j)x}-\sum \limits_{i=1}^n\sum \limits_{j=1}^n a_i a_j e^{(a_i+a_j)x}}{\left(\sum \limits_{i=1}^n e^{a_i x}\right)^2}\\
=&\frac{ \sum \limits_{i=1}^n    \sum \limits_{j=i+1}^n \left( a_i^2 +a_j^2 -2a_i a_j \right) e^{(a_i+a_j)x}  }{\left(\sum \limits_{i=1}^n e^{a_i x}\right)^2} \geq  0.
\end{align}
\end{proof}
Then consider a compound function $h(g(x))$ which is convex if $g(x)$ is convex with respect to $x$.
Hence  -$\ln \mathbb E \left[ e^{-\theta \Phi(n_t)}\right]  $ is concave with respect to $n_t$ and therefore $L(n_t)$ is concave over $n_t$.  From this analysis result, there must exist an optimal $n_t$ that maximizes the function $L(n_t)$.

When taking the decoding error probability into consideration, the monotonicity of $C_E(\theta)$ is the same as $L(n_t)$.
More specifically, $L(n_t)$  can be viewed as a compound function consists of $L_1(x)=-\frac{1}{\theta}\ln(x)$ and $L_2(n_t)= \mathbb E \left\{e^{-\theta m n_d  R (\{\hat \gamma_i\}_1^m, \varepsilon)  } \right \}$.
 Now that the function $L(n_t)=L_1(L_2(n_t))$ is concave and the following condition must hold
\begin{align}
\frac{\partial ^2 L}{\partial n_t^2}&=\frac{\partial ^2 L_1}{\partial L_2^2} \left(\frac{\partial L_2}{\partial n_t}\right)^2+\frac{\partial L_1}{\partial L_2}\frac{\partial ^2 L_2}{\partial n_t^2} \\
&=\frac{1}{L_2(n_t)}\left\{\frac{1}{L_2(n_t)}\left(\frac{\partial L_2}{\partial n_t}\right)^2-\frac{\partial ^2 L_2}{\partial n_t^2}\right\}\leq0.
\end{align}
 For notation convenience, denote
 \begin{align}
 d(n_t)=(1-\varepsilon)L_2(n_t) +\varepsilon,
 \end{align}
 and observe the expression in (\ref{EC-mc-fbl}), $T(n_t)$ can be expressed as
 \begin{align}
 T(n_t)=L_1(d(n_t)).
 \end{align}
 The second order derivative of $T(n_t)$
\begin{align}
\frac{\partial^2 T}{\partial n_t^2}&=\frac{\partial ^2 L_1}{\partial d^2}\left(\frac{\partial d}{\partial n_t}\right)^2+\frac{\partial L_1}{\partial d}\frac{\partial d^2}{\partial n_t^2} \\
&=\frac{1-\varepsilon}{L_2(n_t)+\varepsilon}\left\{ \frac{1-\varepsilon}{L_2(n_t)+\varepsilon}\left(\frac{\partial L_2}{\partial n_t}\right)^2-\frac{\partial^2 L_2}{\partial n_t^2}\right\}\leq0.
\end{align}
{Based on this, we conclude that the effective capacity in (\ref{EC-mc-fbl}) is concave with respect to the pilot length $n_t$.}
\section{Proof of Theorem 3} \label{Theorem3}
Due to $\gamma_0=\frac{p}{\sigma^2}$  and $\sigma^2$ is a constant, we will substitute $p$ by $\gamma_0$ for notation simplicity.
The first order partial  derivative of  $\Psi$ with respect to $\gamma_0$ is as follows
\begin{align}
\frac{\partial \Psi}{\partial \gamma_0}&=\frac{\partial \Psi}{\partial \hat \gamma_i}\frac{\partial \hat \gamma_i}{\partial \gamma_0}\\
&=\frac{\log_2 e}{1+\hat \gamma_i}\frac{n_t\gamma_0\left(n_t\gamma_0+\gamma_0+2\right)}{(1+\gamma_0 + n_t \gamma_0 )^2}x_i \geq 0,
\end{align}
and the second order partial  derivative is shown in (\ref{52})-(\ref{55}).
\begin{align}
\frac{\partial ^2\Psi}{\partial \gamma_0^2}&=\frac{\partial ^2\Psi}{\partial \hat \gamma_i^2}\left(\frac{\partial \hat \gamma_i}{\partial \gamma_0 }\right)^2+\frac{\partial \Psi}{\partial \hat \gamma_i}\frac{\partial^2 \hat \gamma_i}{\partial \gamma_0^2 \label{52}}\\
&=- \frac{\log_2 e}{(1+\hat \gamma_i)^2}\left(\frac{n_t\gamma_0\left(n_t\gamma_0+\gamma_0+2\right)}{(1+\gamma_0 + n_t \gamma_0)^2}x_i\right)^2  +   \frac{\log_2 e}{1+\hat \gamma_i}\frac{2n_t}{(1+\gamma_0 + n_t \gamma_0 )^3}x_i \\
&=- \frac{\log_2 e n_t x_i}{(1+\hat \gamma_i)(1+\gamma_0+n_t\gamma_0)^3}\left[ \frac{(2+\gamma_0+n_t\gamma_0)^2n_t\gamma_0^2x_i}{1+\gamma_0+n_t\gamma_0} -2 \right]\\
&=- \frac{\log_2 e n_t x_i}{(1+\hat \gamma_i)(1+\gamma_0+n_t\gamma_0)^3}\left[ (2+\gamma_0+n_t\gamma_0)^2\hat \gamma_i  -2 \right] > 0, \label{55}
\end{align}
where the last inequality in (\ref{55}) is due  to the assumption $\hat \gamma_i >-3dB$.

Based on the above analysis, $\Psi(p,n_t)$ is concave and monotonically increasing with respect to $p$, we can affirm that the function $L(p)$ is concave and monotonically increasing with respect to $p$. Then we apply the method adopted in the proof of  Theorem 2 and can prove that the effective capacity is also concave and monotonically increasing with respect to $p$.

\section{Proof of Theorem 4} \label{Theorem4}
The first order partial derivative of $L(m, \gamma_0, n_t)$ with respect to $m$ is
\begin{equation}
\begin{aligned}
\frac{\partial L}{\partial m}=& \sqrt{m n_d}Q^{-1}(\varepsilon)\log_2 e \\ &-\frac{m}{\theta}\ln\left\{\mathbb E\left\{ e^{-\theta n_d\log_2(1+\hat \gamma_i)}\right\} \right\},
\end{aligned}
\end{equation}
and the  second order partial derivative of $L(m,p ,n_t)$ with respect to $m$ is
\begin{align}
\frac{\partial^2 L}{\partial m^2}&=-\frac{1}{4} \sqrt\frac{n_d}{m^3}Q^{-1}(\varepsilon)\log_2 e  < 0.
\end{align}

 \begin{equation}
 \begin{aligned}
\Gamma(m, p, \alpha)
=&- \sqrt {mn}  Q^{-1}(\varepsilon) \log_2 e(1-\frac{\alpha}{2})\\ & +\frac{m}{\theta}\ln \left(  \frac{-\theta' n^2 \gamma_0^2  \alpha^2 +(\theta' n^2 \gamma_0^2- n \gamma_0^2+ \gamma_0 n)\alpha +1+\gamma_0}{1+\gamma_0+  n \gamma_0 \alpha} \right).
\label{58}
 \end{aligned}
 \end{equation}

 Given the other parameters,  $L(m, \gamma_0, n_t)$ is concave and is monotonically decreasing with respect $m$, hence the effective capacity $C_E(\theta)$ is monotonically increasing in $m$.
\section{Proof of Theorem 5} \label{proposition1}
Substituting average received SNR $G$ in (\ref{received-SNR}) to (\ref{EC-closed-e}), we can obtain equation (\ref{58}).
According to the definition of  $\Gamma(m, p, \alpha)$,  we can see that the first term of $\Gamma(m,p, \alpha)$ doesn't alter the convexity of $\Gamma(m,p, \alpha)$ over $\alpha$. Then we need to determine the convexity of the latter term.  
The latter term can be viewed as a compound function, where the outer function is negative $\log$ function and it is a convex function and is monotonically decreasing.  The inner function can be simplified as a fractional function as follows
\begin{equation}
{O(\alpha)=\frac{-\theta' n^2 \gamma_0^2  \alpha^2 +(\theta' n^2 \gamma_0^2 - n\gamma_0 ^2+  n\gamma_0)\alpha +1+\gamma_0}{1+\gamma_0+    n\gamma_0 \alpha},}
\end{equation}
the second order derivatives of $O(\alpha)$ can be derived as follows
\begin{equation}
\begin{aligned}
O''(\alpha)=&-\frac{2(1+\gamma_0)(n\gamma_0)^2\left[ (\theta' n_d - 1) \gamma_0+ \theta' (1+\gamma_0)  \right]}{(1+\gamma_0+ n \gamma_0  \alpha)^3}.
\end{aligned}
\end{equation}
Due to the assumption, the inequality $\theta'n_d-1 > 1$ holds and thus $O''(\alpha)< 0$ holds. Therefore,  $O(\alpha)$ is a concave function with respect to $\alpha$. According to the property of the compound function,  the function $\Gamma(\alpha)$ is also concave with respect to $\alpha$.

Due to the fact that the constant terms corresponding with $\varepsilon$ in (\ref{EC-closed}) do not alter the convexity and monotonicity.
Theorem 5 is proved.
\section{Proof of Theorem 6} \label{proposition2}
Here we also substitute $p$ by $\gamma_0$ for simplicity. The second order partial derivative of $\Gamma(m, \gamma_0, \alpha)$ with respect to $\gamma_0$ is shown in (\ref{61})-(\ref{64}), where  the first inequality in (\ref{ineq(1)}) is derived from the fact that $\theta' n_d -1>1$ and the second inequality is derived due to $G>-3$dB.
\begin{align}
\frac{\partial ^2 \Gamma}{\partial \gamma_0^2 }=&\frac{\partial^2 \Gamma}{\partial  G^2}\left (\frac{\partial G}{\partial \gamma_0 }\right)^2+\frac{\partial \Gamma}{\partial  G}\frac{\partial ^2 \ G}{\partial \gamma_0^2 } \label{61}\\
=&-\frac{m}{\theta} \frac{(\theta' n_d -1)^2}{[(\theta' n_d -1)G+1]^2} \left(\frac{[n_t\gamma_0(n_t\gamma_0+\gamma_0+2)]^2}{(1+\gamma_0+n_t\gamma_0)^2}\right)^2 +\frac{m}{\theta} \frac{(\theta'n_d-1)}{(\theta'n_d-1)G+1} \frac{2n_t}{(1+\gamma_0+n_t\gamma_0)^3} \\
=&-\frac{m}{\theta}\frac{(\theta' n_d -1)n_t}{(1+\gamma_0 +n_t\gamma_0 )^3((\theta' n_d -1) G +1)}\times \left(\frac{ \left(n_t\gamma_0+ \gamma_0+2\right)^4(\theta' n_d -1)G}{(\theta' n_d -1) G+1}-2\right)\\
=&-\frac{m}{\theta}\frac{(\theta' n_d -1)n_t}{(1+\gamma_0 +n_t\gamma_0 )^3((\theta' n_d -1) G +1)}\times \left(\frac{ \left(n_t\gamma_0+ \gamma_0+2\right)^4G}{G+\frac{1}{(\theta' n_d -1) }}-2\right) \label{64}\\
<&-\frac{m}{\theta}\frac{(\theta' n_d -1)n_t}{(1+\gamma_0 +n_t\gamma_0 )^3((\theta' n_d -1) G +1)}\times \left(\frac{ \left(n_t\gamma_0+ \gamma_0+2\right)^4G}{G+1}-2\right)\\
<&-\frac{m}{\theta}\frac{(\theta' n_d -1)n_t}{(1+\gamma_0 +n_t\gamma_0 )^3((\theta' n_d -1) G +1)}\times \left(\frac{ 16G}{G+1}-2\right) \label {ineq(1)}<0.
\end{align}

Similarly, due to the fact that the constant terms corresponding with $\varepsilon$ in (\ref{EC-closed}) do not alter the convexity and monotonicity, Theorem 6 is proved.

\end{appendices}
\bibliographystyle{IEEEtran}
\bibliography{myref}

\begin{thebibliography}{10}
\providecommand{\url}[1]{#1}
\csname url@samestyle\endcsname
\providecommand{\newblock}{\relax}
\providecommand{\bibinfo}[2]{#2}
\providecommand{\BIBentrySTDinterwordspacing}{\spaceskip=0pt\relax}
\providecommand{\BIBentryALTinterwordstretchfactor}{4}
\providecommand{\BIBentryALTinterwordspacing}{\spaceskip=\fontdimen2\font plus
\BIBentryALTinterwordstretchfactor\fontdimen3\font minus
  \fontdimen4\font\relax}
\providecommand{\BIBforeignlanguage}[2]{{%
\expandafter\ifx\csname l@#1\endcsname\relax
\typeout{** WARNING: IEEEtran.bst: No hyphenation pattern has been}%
\typeout{** loaded for the language `#1'. Using the pattern for}%
\typeout{** the default language instead.}%
\else
\language=\csname l@#1\endcsname
\fi
#2}}
\providecommand{\BIBdecl}{\relax}
\BIBdecl

\bibitem{8472907}
M.~Bennis, M.~Debbah, and H.~V. Poor, ``Ultrareliable and low-latency wireless
  communication: Tail, risk, and scale,'' \emph{Proc. {IEEE}}, vol. 106,
  no.~10, pp. 1834--1853, 2018.

\bibitem{8705373}
P.~Popovski, C.~Stefanovic, J.~J. Nielsen, E.~D. Carvalho, and A.~S. Bana,
  ``Wireless access in ultra-reliable low-latency communication ({URLLC}),''
  \emph{IEEE Trans. Commun.}, vol.~67, no.~8, pp. 1--1, 2019.

\bibitem{8469808}
H.~Chen, R.~Abbas, P.~Cheng, M.~Shirvanimoghaddam, W.~Hardjawana, W.~Bao,
  Y.~Li, and B.~Vucetic, ``Ultra-reliable low latency cellular networks: Use
  cases, challenges and approaches,'' \emph{IEEE Commun. Mag.}, vol.~56,
  no.~12, pp. 119--125, 2018.

\bibitem{8452975}
C.~Li, C.~Li, K.~Hosseini, S.~B. Lee, J.~Jiang, W.~Chen, G.~Horn, T.~Ji, J.~E.
  Smee, and J.~Li, ``{5G}-based systems design for tactile internet,''
  \emph{Proc.IEEE}, vol. 107, no.~2, pp. 307--324, 2019.

\bibitem{7529226}
G.~Durisi, T.~Koch, and P.~Popovski, ``Toward massive, ultrareliable, and
  low-latency wireless communication with short packets,'' \emph{Proc. {IEEE}},
  vol. 104, no.~9, pp. 1711--1726, 2016.

\bibitem{5452208}
Y.~Polyanskiy, H.~V. Poor, and S.~Verdu, ``Channel coding rate in the finite
  blocklength regime,'' \emph{IEEE Trans. Inf. Theory}, vol.~56, no.~5, pp.
  2307--2359, 2010.

\bibitem{8403963}
H.~Ji, S.~Park, J.~Yeo, Y.~Kim, J.~Lee, and B.~Shim, ``Ultra-reliable and
  low-latency communications in {5G} downlink: Physical layer aspects,''
  \emph{IEEE Wireless Commun.}, vol.~25, no.~3, pp. 124--130, 2018.

\bibitem{6802432}
W.~Yang, G.~Durisi, T.~Koch, and Y.~Polyanskiy, ``Quasi-static multiple-antenna
  fading channels at finite blocklength,'' \emph{IEEE Trans. Inf. Theory},
  vol.~60, no.~7, pp. 4232--4265, 2014.

\bibitem{7362178}
G.~Durisi, T.~Koch, J.~Östman, Y.~Polyanskiy, and W.~Yang, ``Short-packet
  communications over multiple-antenna rayleigh-fading channels,'' \emph{IEEE
  Trans. Commun.}, vol.~64, no.~2, pp. 618--629, 2016.

\bibitem{7463506}
S.~Xu, T.~Chang, S.~Lin, C.~Shen, and G.~Zhu, ``Energy-efficient packet
  scheduling with finite blocklength codes: Convexity analysis and efficient
  algorithms,'' \emph{IEEE Trans. Wireless Commun.}, vol.~15, no.~8, pp.
  5527--5540.

\bibitem{8345745}
X.~Sun, S.~Yan, N.~Yang, Z.~Ding, C.~Shen, and Z.~Zhong, ``Short-packet
  downlink transmission with non-orthogonal multiple access,'' \emph{IEEE
  Trans. Wireless Commun.}, vol.~17, no.~7, pp. 4550--4564, 2018.

\bibitem{8259329}
Y.~Hu, M.~Serror, K.~Wehrle, and J.~Gross, ``Finite blocklength performance of
  cooperative multi-terminal wireless industrial networks,'' \emph{IEEE Trans.
  Veh. Technol.}, vol.~67, no.~7, pp. 5778--5792, 2018.

\bibitem{6888474}
B.~Makki, T.~Svensson, and M.~Zorzi, ``Finite block-length analysis of the
  incremental redundancy {HARQ},'' \emph{IEEE Wireless Commun. Lett.}, vol.~3,
  no.~5, pp. 529--532, 2014.

\bibitem{7996416}
M.~Mousaei and B.~Smida, ``Optimizing pilot overhead for ultra-reliable
  short-packet transmission,'' in \emph{2017 IEEE Int. Conf. Commun. (ICC)},
  2017.

\bibitem{8885564}
Y.~Zhu, Y.~Hu, Z.~Chang, and A.~Schmeink, ``Throughput maximization of
  low-latency communication with imperfect {CSI} in finite blocklength
  regime,'' in \emph{2019 IEEE Wireless Commun. Netw. Conf. (WCNC)}, 2019.

\bibitem{9013958}
J.~Cao, X.~Zhu, Y.~Jiang, Y.~Liu, and F.~Zheng, ``Joint block length and pilot
  length optimization for {URLLC} in the finite block length regime,'' in
  \emph{2019 IEEE Global Commun. Conf. (GLOBECOM)}, 2019.

\bibitem{9120794}
J.~Cheng, C.~Shen, and S.~Xia, ``Robust {URLLC} packet scheduling of {OFDM}
  systems,'' in \emph{2020 IEEE Wireless Commun. Netw. Conf. (WCNC)}, 2020.

\bibitem{9044874}
H.~Ren, C.~Pan, Y.~Deng, M.~Elkashlan, and A.~Nallanathan, ``Joint pilot and
  payload power allocation for massive-{MIMO}-enabled {URLLC} {IIoT}
  networks,'' \emph{IEEE J. Sel. Areas Commun.}, vol.~38, no.~5, pp. 816--830,
  May 2020.

\bibitem{1210731}
D.~Wu and R.~Negi, ``Effective capacity: a wireless link model for support of
  quality of service,'' \emph{IEEE Trans. Wireless Commun.}, vol.~2, no.~4, pp.
  630--643, 2003.

\bibitem{Amjad2019}
M.~Amjad, L.~Musavian, and M.~Rehmani, ``Effective capacity in wireless
  networks: A comprehensive survey,'' \emph{IEEE Commun. Surveys Tuts.},
  vol.~21, no.~4, pp. 3007--3038, 2019.

\bibitem{8913775}
L.~Zhang, Y.~Yang, X.~Li, J.~Chen, and Y.~Chi, ``Effective capacity in
  cognitive radio networks with relay and primary user emulator,'' \emph{China
  Commun.}, vol.~16, no.~11, pp. 130--145, 2019.

\bibitem{6932489}
H.~Al-Zubaidy, J.~Liebeherr, and A.~Burchard, ``Network-layer performance
  analysis of multihop fading channels,'' \emph{IEEE/ACM Trans. Netw.},
  vol.~24, no.~1, pp. 204--217, 2016.

\bibitem{GursoyDec.2013}
M.~C. Gursoy, ``Throughput analysis of buffer-constrained wireless systems in
  the finite blocklength regime,'' \emph{EURASIP J. Wireless Commun. and
  Netw.}, Dec. 2013.

\bibitem{8402240}
Y.~Hu, O.~Mustafa, G.~M. Cenk, and S.~Anke, ``Optimal power allocation for
  {QoS}-constrained downlink multi-user networks in the finite blocklength
  regime,'' \emph{IEEE Trans. Wireless Commun.}, vol.~17, no.~9, pp.
  5827--5840, 2018.

\bibitem{8645712}
M.~Shehab, H.~Alves, and M.~Latva-Aho, ``Effective capacity and power
  allocation for machine-type communication,'' \emph{IEEE Trans. Veh.
  Technol.}, vol.~68, no.~4, pp. 4098--4102, 2019.

\bibitem{8543235}
J.~Choi, ``An effective capacity-based approach to multi-channel low-latency
  wireless communications,'' \emph{IEEE Trans. Commun.}, vol.~67, no.~3, pp.
  2476--2486, 2019.

\bibitem{QiaoAug.2019}
D.~Qiao, M.~C. Gursoy, and S.~Velipasalar, ``Throughput-delay tradeoffs with
  finite blocklength coding over multiple coherence blocks,'' \emph{IEEE Trans.
  Commun.}, vol.~67, no.~8, pp. 5892--5904, 2019.

\bibitem{8638930}
C.~Xiao, J.~Zeng, W.~Ni, X.~Su, R.~P. Liu, T.~Lv, and J.~Wang, ``Downlink
  {MIMO-NOMA} for ultra-reliable low-latency communications,'' \emph{IEEE J.
  Sel. Areas Commun.}, vol.~37, no.~4, pp. 780--794, 2019.

\bibitem{8649644}
C.~Xiao, J.~Zeng, W.~Ni, R.~P. Liu, X.~Su, and J.~Wang, ``Delay guarantee and
  effective capacity of downlink {NOMA} fading channels,'' \emph{IEEE J. Sel.
  Topics Signal Process.}, vol.~13, no.~3, pp. 508--523, Jun. 2019.

\bibitem{8421266}
S.~Schiessl, H.~Al-Zubaidy, M.~Skoglund, and J.~Gross, ``Delay performance of
  wireless communications with imperfect {CSI} and finite-length coding,''
  \emph{IEEE Trans. Commun.}, vol.~66, no.~12, pp. 6527--6541, 2018.

\bibitem{8640115}
S.~Schiessl, J.~Gross, M.~Skoglund, and G.~Caire, ``Delay performance of the
  multiuser miso downlink under imperfect {CSI} and finite-length coding,''
  \emph{IEEE J. Sel. Areas Commun.}, vol.~37, no.~4, pp. 765--779, 2019.

\bibitem{5466522}
G.~Caire, N.~Jindal, M.~Kobayashi, and N.~Ravindran, ``Multiuser {MIMO}
  achievable rates with downlink training and channel state feedback,''
  \emph{IEEE Trans. Inf. Theory}, vol.~56, no.~6, pp. 2845--2866, 2010.

\bibitem{1193803}
B.~Hassibi and B.~M. Hochwald, ``How much training is needed in
  multiple-antenna wireless links?'' \emph{IEEE Trans. Inf. Theory}, vol.~49,
  no.~4, pp. 951--963, 2003.

\bibitem{5205834}
Y.~Polyanskiy, H.~V. Poor, and S.~Verdu, ``Dispersion of gaussian channels,''
  in \emph{2009 IEEE Int. Symp. Inf. Theory{ (ISIT)}}, pp. 2204--2208, {IEEE},
  2009.

\bibitem{Schiessl2015}
S.~Schiessl, J.~Gross, and H.~Al-Zubaidy, ``Delay analysis for wireless fading
  channels with finite blocklength channel coding,'' in \emph{Proc. ACM MSWiM},
  2015.

\bibitem{Chiccoli1990}
C.~Chiccoli, S.~Lorenzutta, and G.~Maino, ``Recent results for generalized
  exponential integrals,'' \emph{Comput. Math. Appl.}, vol.~19, no.~5, p.
  21–29, 1990.

\end{thebibliography}
\end{document}